\documentclass[journal,twoside]{IEEEtran}
\usepackage{cite}
\usepackage{amsmath,amssymb,amsfonts}
\usepackage{amsthm}
\newtheorem{theorem}{Theorem}
\newtheorem{proposition}{Proposition}
\usepackage{graphicx}
\usepackage{textcomp}
\usepackage{xcolor}
\usepackage{algorithm}
\usepackage[noend]{algpseudocode}
\usepackage{makecell}
\usepackage{array}
\usepackage{url}
\usepackage{subcaption}
\usepackage[breaklinks=true,hidelinks]{hyperref}

\def\BibTeX{{\rm B\kern-.05em{\sc i\kern-.025em b}\kern-.08em
    T\kern-.1667em\lower.7ex\hbox{E}\kern-.125emX}}
\begin{document}

\title{SbDN: Source-based TSN-Grade Deterministic Networking using Commodity Switches}

\author{Mohammadparsa Karimi, Majid Nabi, Andrew Nelson, Kees Goossens, \IEEEmembership{Member, IEEE}, and Twan Basten, \IEEEmembership{Senior Member, IEEE}
\thanks{This work received funding from the European Chips Joint Undertaking under Framework Partnership Agreement No.~101139789 (HAL4SDV).}
\thanks{The authors are with the Electronic Systems (ES) group, Department of Electrical Engineering, Eindhoven University of Technology (TU/e), 5612~AZ Eindhoven, The~Netherlands (e-mail: \{m.karimi, m.nabi, a.t.nelson, k.g.w.goossens, a.a.basten\}@tue.nl).}}
\maketitle
\markboth{ }%
{Karimi \MakeLowercase{\textit{et al.}}: SbDN: Source-based TSN-Grade Deterministic Networking using Commodity Switches}

\begin{abstract}
Deterministic networking is essential for safety-critical applications in automotive, industrial, and aerospace systems, where bounded end-to-end latency must be guaranteed for time-critical traffic. Time-Sensitive Networking (TSN) provides the mechanisms to achieve such guarantees, but its deployment requires expensive TSN-capable switches at every hop and complex per-switch configuration that hinders runtime reconfiguration. This paper presents SbDN, a Multi-Agent Source-based architecture that achieves TSN-grade determinism using commodity Ethernet switches. SbDN moves all scheduling intelligence to a centralized controller composed of three cooperating agents and enforces the computed configurations exclusively at the source endpoints, leaving switches as simple forwarding elements. We propose two methods: Temporal Network Partitioning (TNP), which provides strict temporal isolation on pure FIFO switches, and Traffic Prioritization (TP), which leverages strict-priority queuing at switches to enable work-conserving best-effort traffic. Both methods are formally proven to guarantee that all admitted time-critical flows meet their end-to-end deadlines. Evaluation across 40 benchmark configurations on two topologies shows that TNP and TP achieve 100\% admission of time-critical traffic in every scenario, with scheduling times in the low-millisecond range suitable for safe runtime reconfiguration. Compared to a standard TSN baseline, SbDN delivers superior time-critical latency at a fraction of the switch infrastructure cost, while offering competitive best-effort throughput through the choice between the two methods.
\end{abstract}


\section{Introduction}\label{sec:intro}

Deterministic networking has become a fundamental requirement across a growing range of industries where communication failures or excessive delays can compromise safety, reliability, or operational correctness. In automotive systems, distributed sensors, controllers, and actuators exchange time-critical data to enable functions such as autonomous driving, active safety, and powertrain control, all of which demand bounded end-to-end latencies in the order of microseconds to milliseconds \cite{Automotive}. In industrial automation, closed-loop control of robotic arms, conveyor systems, and process plants relies on periodic sensor-actuator communication that must meet strict deadlines to maintain stability and precision \cite{Industry}. In aerospace and avionics, fly-by-wire control and mission-critical telemetry require deterministic delivery with guaranteed fault tolerance \cite{Aerospace}. In each of these domains, the network must provide not only high throughput but also predictable, bounded timing behavior that standard best-effort Ethernet cannot guarantee.

Time-Sensitive Networking (TSN), a suite of standards developed by the IEEE 802.1 Task Group, has emerged as the leading technology for providing deterministic communication over standard Ethernet. TSN extends conventional Ethernet with a set of mechanisms that collectively enable bounded latency, temporal isolation between traffic classes, and coexistence of time-critical and best-effort traffic on a shared physical network \cite{TSN}. Among these mechanisms, the Time-Aware Shaper (TAS, IEEE 802.1Qbv) \cite{IEEEQbv} provides deterministic transmission windows through gate-controlled scheduling at each switch port, the Credit-Based Shaper (CBS, IEEE 802.1Qav) \cite{IEEEQav} regulates the transmission rate of individual traffic classes to prevent bandwidth starvation, frame preemption (IEEE 802.1Qbu/802.3br) \cite{IEEEQbu} allows high-priority frames to interrupt ongoing lower-priority transmissions to reduce worst-case latency, and IEEE 802.1AS \cite{IEEEAS} ensures network-wide time synchronization through the generalized Precision Time Protocol (gPTP). Additional standards within the TSN suite address further aspects such as stream reservation, traffic policing, and fault tolerance. Together, these mechanisms allow TSN to support mixed-criticality workloads in which safety-critical traffic and lower-priority data streams share the same physical infrastructure while each meets its respective communication requirements \cite{TSNG}. In this work, we focus on two such classes that capture the essential mixed-criticality trade-off: Time-Critical (TC) traffic, which carries delay-sensitive data subject to hard end-to-end deadlines, and Best-Effort (BE) traffic, which has no timing guarantees but should achieve high throughput on the remaining bandwidth.

Despite its technical capabilities, deploying TSN in practice comes at a significant cost. TSN-capable switches that support mechanisms such as TAS and CBS are substantially more expensive than standard Ethernet switches, as they require specialized hardware for gate-controlled scheduling, per-port credit shaping, and hardware-assisted time synchronization. Moreover, adopting TSN in an existing network is not an incremental upgrade; it typically requires replacing the entire switching infrastructure with TSN-enabled devices, since every switch on the path of a time-critical flow must participate in the scheduling. This cost barrier limits the practical adoption of TSN, particularly in cost-sensitive industries \cite{Cost}.

Beyond cost, the operational complexity of TSN configuration presents a further challenge. Deterministic behavior in TSN is achieved entirely through the correct configuration and scheduling parameters at every switch along the path of each time-critical flow \cite{TSNscheduling}. When network conditions change, for instance due to the addition of new flows, removal of existing ones, or shifts in traffic patterns, the configuration must be recomputed and redeployed across all affected switches. However, updating the configuration of active switches during operation risks disrupting the timing guarantees of all ongoing flows, since the transition from one schedule to another cannot in general be performed atomically across multiple devices. This instability during reconfiguration makes online adaptation difficult in practice, and many deployments avoid runtime reconfiguration altogether, restricting schedule updates to periods when the system is idle. This fundamentally limits the flexibility of TSN-based systems in environments where traffic patterns evolve dynamically \cite{Reconfiguration}.

In this work, we introduce SbDN, Source-based TSN-Grade Deterministic Networking without requiring TSN-capable switches. Instead of deploying scheduling and shaping mechanisms at every switch, SbDN moves all scheduling intelligence to a centralized software controller composed of cooperating agents and enforces the computed configuration at the source endpoints only. Switches remain commodity forwarding elements with no scheduling functionality. Throughout this paper, a \emph{commodity switch} is a standard store-and-forward Ethernet switch without TSN scheduling or shaping support. We propose two methods in this work, one of which requires only FIFO forwarding at the switches, while the other requires standard strict-priority queuing; both require the endpoints to share a common notion of time. The main contributions of this work are as follows:

\begin{enumerate}
\item We define a formal system model and problem formulation for source-based deterministic networking using TSN functionalities enforced exclusively at the endpoints.
\item We propose two architectural scenarios that achieve deterministic communication using different combinations of source-side mechanisms and switch capabilities, offering distinct trade-offs between switch complexity and scheduling complexity.
\item We provide formal guarantees on time-critical latency, including a proof that all TC frames meet their end-to-end deadlines under the proposed configurations.
\item We evaluate SbDN through extensive simulation, comparing its performance against both an Ethernet-based solution and standard TSN baseline.
\item We release the complete implementation as open-source software to support reproducibility and further research.\footnote{The software is available via the TU/e ES GitHub repository (https://github.com/TUE-EE-ES/SbDN).}
\end{enumerate}

The key idea of SbDN is to replace per-switch scheduling hardware with a software-defined central controller that computes all configurations and pushes them to the source endpoints before starting execution. SbDN organizes time into a repeating scheduling cycle of fixed duration, within which all periodic TC frames are placed and BE traffic is regulated. 
The controller is composed of three cooperating agents. The TC Scheduling Agent assigns collision-free release times to all time-critical frames, ensuring that no two TC frames occupy the same link simultaneously anywhere in the network. The Network Partitioning Agent constructs a global gate control list that partitions the scheduling cycle into TC and BE phases, providing temporal separation between traffic classes without requiring any TSN support at the switches. The BE Shaping Agent configures a credit-based idle slope at each endpoint to regulate best-effort transmission rates and bound queue buildup at intermediate switches. We propose two methods that combine these agents in different ways. Temporal Network Partitioning (TNP) employs all three agents and targets networks with pure FIFO commodity switches, achieving determinism through strict temporal separation enforced at the sources. Traffic Prioritization (TP) employs only the TC Scheduling Agent and the BE Shaping Agent, relying on commodity switches with strict-priority queuing to handle TC-versus-BE contention, which simplifies the scheduling problem and makes BE traffic work-conserving.

The remainder of this paper is organized as follows. Section~\ref{sec:example} presents an illustrative example that motivates the source-based approach. Section~\ref{sec:related_work} reviews related work. The system model and problem definition are presented in Section~\ref{sec:sysmodel}. Section~\ref{sec:SbDN} details the SbDN architecture, including both scenarios and their formal guarantees. Section~\ref{sec:eval} presents the evaluation and results. Section~\ref{sec:conclusions} concludes the paper and outlines directions for future work.

\section{Illustrative Example}\label{sec:example}

Although the SbDN architecture is domain-agnostic, we illustrate its operation through a concrete automotive example that highlights the key ideas before the formal treatment in subsequent sections. The high-level application scenario is adapted from~\cite{IlustExample}. 

\subsection{Application Scenario}

Consider an Advanced Driver Assistance System (ADAS) in which three sensors transmit periodic data to a central compute Electronic control unit (ECU): a LiDAR sensor operating at 10~Hz (period $T = 100$~ms), a camera operating at 20~Hz (period $T = 50$~ms), and an Inertial Measurement Unit (IMU) operating at 100~Hz (period $T = 10$~ms). Each sensor stream is time-critical and must arrive at the ECU within a bounded deadline. In addition to these TC flows, the network also carries BE traffic.

The network uses the topology shown in Figure~\ref{fig:ilustrative_topology}: four endpoints connected through two switches. Three endpoints ($v_1$ for LiDAR, $v_2$ for camera, $v_3$ for IMU) serve as sensors and the fourth ($v_4$) is the ADAS compute ECU. All links operate at 1~Gbps with a switch processing delay of 1~µs. For simplicity, we assume that the LiDAR and camera produce frames of 1500~B (transmission time 12~µs per link), while the IMU produces smaller frames of 250~B (transmission time 2~µs per link).\footnote{Realistic camera and LiDAR frame sizes may result in several 1000s of 1500 B max-sized Ethernet packets per frame; this would complicate the illustrative example. This simplifying assumption does not fundamentally impact the examples and proposed techniques.} The LiDAR and camera connect to switch $s_1$, the IMU connects to switch $s_2$, and the ECU also connects to $s_2$. All three sensor flows converge on the link $e_5$ connecting $s_2$ to the ECU, making it the shared bottleneck. Within a single 100~ms scheduling cycle, the TC Scheduling Agent must place 13 frames: 1 LiDAR, 2 camera, and 10 IMU frames. The chosen periods all divide the cycle duration exactly, so the same frame count repeats in every cycle.

\begin{figure}[t]
\centering
\includegraphics[width=0.9\columnwidth]{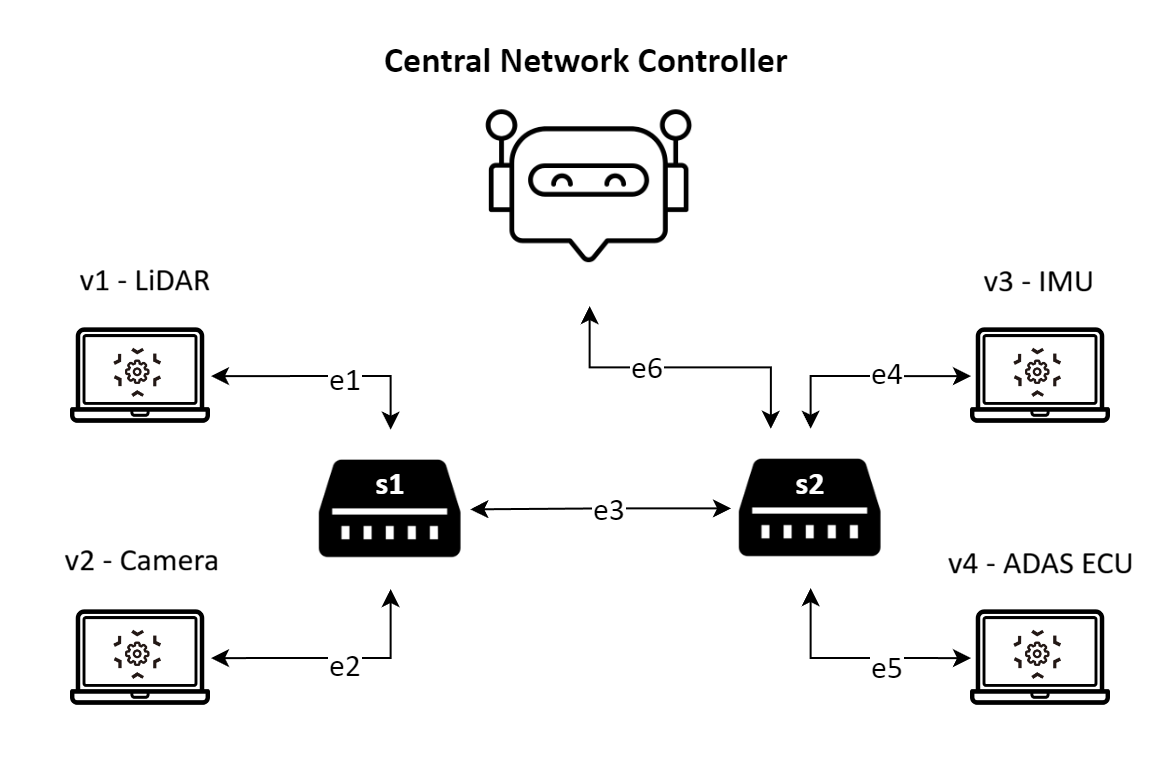}
\caption{Illustrative example topology}
\label{fig:ilustrative_topology}
\end{figure}

\subsection{TNP Method}

Temporal Network Partitioning assigns collision-free release times to all 13 TC frames and partitions the cycle into dedicated TC and BE phases enforced by a global Gate Control List (GCL) at every endpoint. A GCL, defined in IEEE 802.1Qbv and detailed in Section~\ref{sec:tas_model}, specifies for each segment of a repeating cycle which traffic classes are permitted to transmit; here, we use it at the source endpoints rather than at the switch egress ports, where standard TSN applies, in general port-specific, GCLs. The cycle is called \emph{global} because, in the current version of SbDN, the Network Partitioning Agent constructs a single GCL that is enforced identically at every endpoint of the network; more compact per-endpoint GCLs are an alternative construction, discussed in Section~\ref{sec:tas_agent}. SbDN requires all endpoints to share a common notion of time. We assume PTP-based synchronization, which can be provided either by commodity switches with PTP support or by software-based solutions; the synchronization requirements are discussed in Section~\ref{sec:SbDN}. Figure~\ref{fig:illust_tnp_perlink} shows the resulting frame placement across all five links over the full 100~ms cycle. Each frame appears at its scheduled position: the 10 IMU frames at regular 10~ms intervals on links $e_4$ and $e_5$, the two camera frames at 0 and 50~ms on links $e_2$, $e_3$, and $e_5$, and the single LiDAR frame near the start of the cycle on links $e_1$, $e_3$, and $e_5$.\footnote{For realistic camera/LiDAR frame sizes, these individual camera and LiDAR frames turn into bursts of frames transmitted consecutively, not fundamentally changing the schedule. The only effect is that the first and sixth TC phases in Fig.~\ref{fig:illust_tnp_gcl}, explained below, become larger.} No two frames overlap on any link. The top panel of the figure is schematic and not to scale: the actual frame durations are in the order of microseconds while the cycle spans 100~ms. So each frame is plotted at its source release time on every link of its route. To make the per-hop pipelining visible, the bottom panel zooms into the first TC cluster at the microsecond scale. It shows how the LiDAR frame propagates from one link to the next: after being transmitted on $e_1$, it is processed by switch $s_1$ and then transmitted on $e_3$, and similarly forwarded by $s_2$ onto $e_5$. Each hop adds the transmission time of the frame plus a small switch processing delay, totaling roughly 13~µs per hop in this example. The same pipelining applies to the camera frame, which traverses the same shared link $e_3$ and is packed back-to-back behind the LiDAR.

\begin{figure}[t]
\centering
\includegraphics[width=\columnwidth]{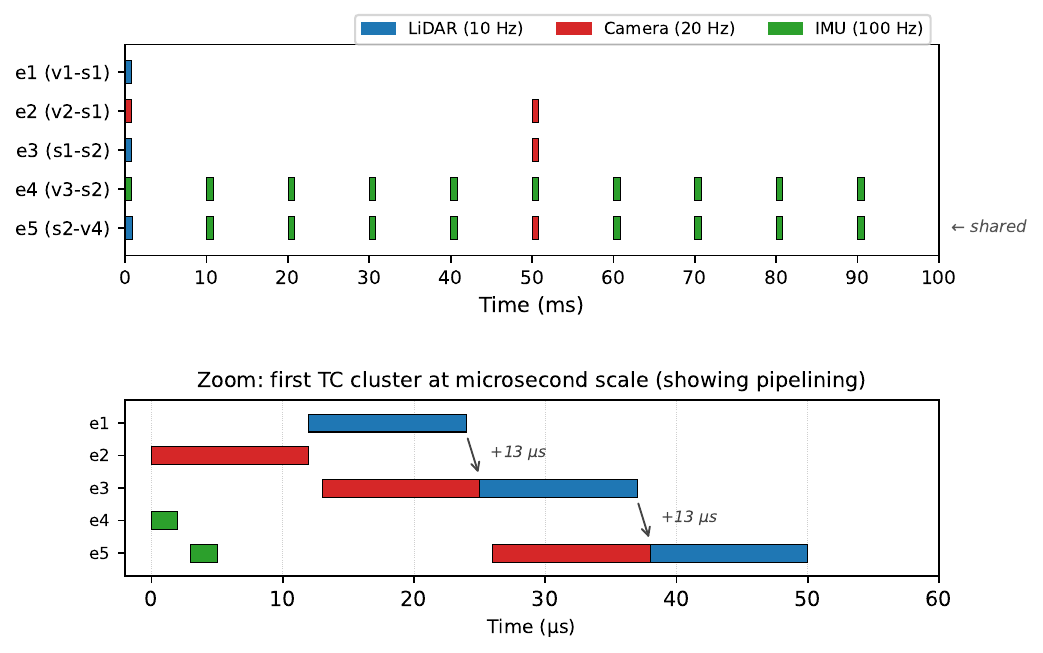}
\caption{TNP: collision-free placement of all 13 TC frames. Top: per-link timelines over one 100~ms cycle (schematic). Bottom: zoom into the first TC cluster at the microsecond scale.}
\label{fig:illust_tnp_perlink}
\end{figure}

Based on this placement, the Network Partitioning Agent constructs a global GCL that alternates between TC phases and BE phases. A guard band is inserted before each TC phase to ensure that all in-flight BE frames have drained from the network. Figure~\ref{fig:illust_tnp_gcl} shows the resulting cycle structure: 10 short TC phases (one for each 10~ms boundary where IMU frames are released, with the camera and LiDAR frames absorbed into the clusters at 0 and 50~ms), separated by BE phases during which best-effort traffic may use the network.

\begin{figure}[t]
\centering
\includegraphics[width=\columnwidth]{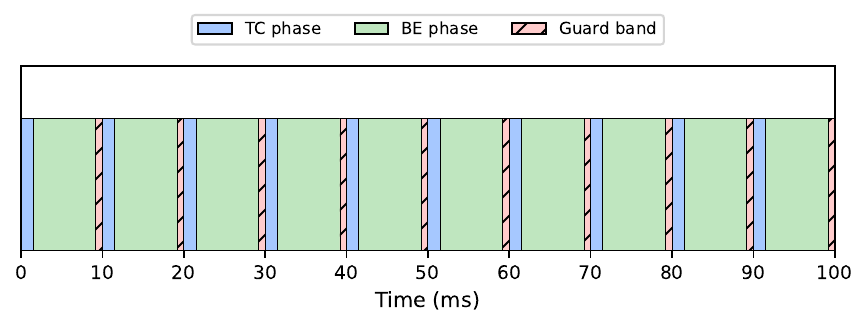}
\caption{TNP: visualization of the GCL cycle, enforced identically at every endpoint, with 10 TC phases, BE phases, and guard bands over one 100~ms cycle (schematic, not to scale).}
\label{fig:illust_tnp_gcl}
\end{figure}

From the per-link placement in Fig.~\ref{fig:illust_tnp_perlink}, the Network Partitioning Agent builds a single global GCL that is identical at every endpoint. It collects the TC reservations from all links onto one timeline and marks every interval that carries a TC frame anywhere in the network as a TC phase, leaving the remaining intervals as BE phases. Each endpoint enforces this same GCL: it transmits its own TC frames at their scheduled release times during the TC phases and confines its BE traffic to the BE phases. Any given source transmits in only some of the TC phases; the LiDAR source $v_1$, for example, sends a single TC frame near the start of the cycle and is idle during the later ones. In a TC phase where it does not transmit, $v_1$ simply keeps its BE gate closed, while its TC gate stays open harmlessly, since it has no TC frame scheduled then and the closed BE gate prevents any leakage. A per-endpoint GCL could drop the TC segments in which other endpoints have frames scheduled and be more compact, but it would not enlarge the BE time in this case: every flow in this topology shares the bottleneck link $e_5$, so every source must keep its BE gate closed during every TC phase.

Because the switches are pure FIFO forwarders, the temporal separation enforced by the global GCL is the sole mechanism that prevents BE traffic from interfering with TC frames. No TSN hardware is required at the switches.

\subsection{TP Method}

Traffic Prioritization achieves the same deadline guarantees using commodity switches with strict-priority queuing instead of temporal partitioning. Because switches with strict priority always serve TC frames before BE frames, no GCL or guard band is needed. Figure~\ref{fig:illust_tp_perlink} shows the resulting frame placement. At the 100~ms scale (top panel), the placement looks nearly identical to TNP; the difference becomes visible only at the microsecond scale (bottom panel), where the per-hop offset of the LiDAR frame grows from 13~µs in TNP to 25~µs in TP. This inflation reflects the worst-case BE blocking that an initial TC frame in a burst of consecutive TC frames may experience at each switch if a BE frame has just begun transmission on the output link: because frames are transmitted in full once started, a newly arriving TC frame may have to wait for a full-size BE frame to finish before being served, adding up to 12~µs per switch on top of the transmission and processing delays. This blocking is accounted for in the scheduling and is critical for the formal deadline guarantees. In return, BE traffic transmits continuously throughout the cycle without being confined to designated phases, as indicated by the shaded background in both panels.

\begin{figure}[t]
\centering
\includegraphics[width=\columnwidth]{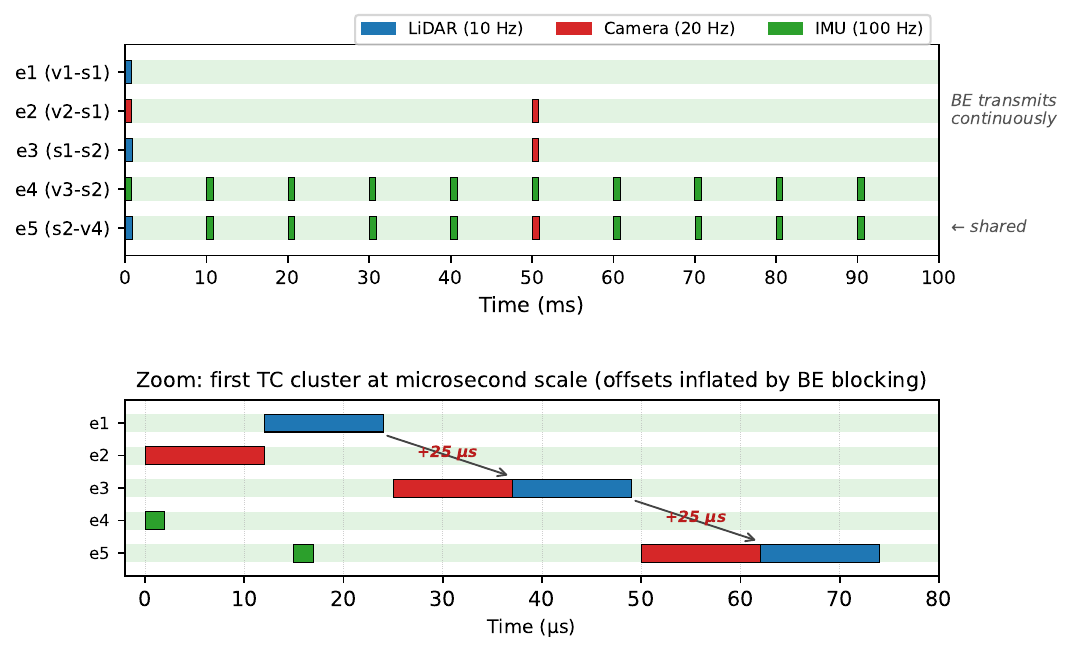}
\caption{TP: TC frame placement with offsets inflated by worst-case BE blocking. Top: per-link timelines over one 100~ms cycle (schematic). Bottom: zoom into the first TC cluster at the microsecond scale. The shaded background indicates BE traffic may transmit at any time outside the intervals occupied by TC frames.}
\label{fig:illust_tp_perlink}
\end{figure}

\subsection{TNP-TP Comparison}

Both methods guarantee that every TC frame meets its deadline, but they differ in how this guarantee is achieved. TNP enforces strict temporal separation between traffic classes using a source-side GCL on pure FIFO switches, while TP relies on strict-priority queuing at the switches to resolve contention, resulting in a simpler source configuration and higher BE utilization. In both cases, all scheduling intelligence resides in the central controller and all enforcement occurs at the source endpoints. The switches do not require TSN functionality. The formal system model, algorithms, and proofs that underpin these results are presented in Sections~\ref{sec:sysmodel}~and~\ref{sec:SbDN}.

\section{Related Work}\label{sec:related_work}

Research on TSN configuration has progressed along several distinct directions. We organize this section around three lines of work that frame our contribution: scheduling of time-critical traffic, joint scheduling of mixed-criticality traffic, and recent efforts to provide deterministic guarantees over commodity switches. We then position SbDN relative to these directions and motivate our choice of baselines.

\subsection{Time-Critical Scheduling}

A large body of TSN work addresses the synthesis of GCLs for the TAS at every TSN switch port to guarantee bounded latency for time-critical (TC) flows. Exact methods formulate the problem as Integer Linear Programs \cite{Schweissguth2020} or Satisfiability Modulo Theories instances \cite{Craciunas2016}, producing optimal schedules at the cost of runtimes that grow rapidly with the network size and the number of flows. Heuristic approaches trade optimality for scalability. D\"urr and Nayak \cite{Durr2016} propose a no-wait scheduling formulation that minimizes flow time, and Move-to-Front tabling \cite{Jin2020} compresses GCLs by packing transmissions back-to-back. Beyond classical heuristics, several recent works apply learning-based techniques to TAS scheduling, including Deep Reinforcement Learning approaches that adapt schedules at runtime \cite{Karimi2025DRL,He2023DeepScheduler,Roberty2024}. These works share the same underlying assumption: every switch in the network is TSN-capable and supports per-port TAS hardware. Their contribution lies in algorithmic efficiency rather than in the architecture of where scheduling is enforced.

Beyond Ethernet, the systems-on-chip community, where ultra-low communication latency is essential, developed the idea to create and align slot reservations across networks-on-chip to provide low latency and latency guarantees for dedicated communication streams \cite{Goossens2005Aethereal,Psarras2016PhaseNoC,Kasapaki2016Argo}. Our SbDN approach develops similar solutions for time-critical communication over commodity Ethernet networks, building on the TSN TAS concept, combined with traffic shaping of best-effort traffic and optional strict-priority queuing.

\subsection{Mixed-Criticality Scheduling}

A second line of work addresses the joint configuration of multiple traffic classes coexisting on the same network, typically TC, Audio Video Bridging (AVB), and BE. Gavrilu\c{t} and Pop \cite{Gavrilut2018} were among the first to point out that synthesizing GCLs for TC traffic in isolation can render AVB traffic unschedulable, and they proposed a GRASP-based metaheuristic that takes AVB worst-case delays into account when selecting TC schedules. Guo et al.\ \cite{Guo2024} observe that the static assignment of flows to TSN traffic classes can cause severe queuing delays under uneven load, and propose dependency-aware priority adjustment algorithms that dynamically reassign flows across TT, AVB-A, and AVB-B classes based on link-overlap conditions and flow features. EAST~\cite{Mateu2026} considers the scheduling of scheduled traffic (ST) and AVB, which are both time critical. EAST addresses the inefficiency of iterative feedback loops between ST schedule synthesis and AVB schedulability analysis by performing a worst-case response time analysis to derive, for each AVB stream, the maximum ST interference it can tolerate without missing its deadline. This bound is then translated into a sliding-window constraint imposed on the ST scheduler, enabling single-pass synthesis that guarantees schedulability of both classes and reducing scheduling time by several orders of magnitude compared to prior methods. HERMES \cite{Bujosa2022} schedules TC frames link-by-link from destination to source using multiple TC queues and produces GCLs that coexist with AVB and BE traffic on the same switch ports, achieving scheduling times below 10~ms while preserving compatibility with the CBS shaping of lower-priority classes. These works assume TSN-capable hardware throughout the network and do not address the cost of deployment. Although in this paper, we consider only TC and BE traffic, both the TNP and TP approaches can be extended to multiple classes with timing constraints as in the discussed papers by considering all these classes as TC traffic. A dedicated scheduler could then be integrated in the proposed TC Scheduling Agent to compute schedules for multiple TC flows from different classes that then serve as a basis for TNP or TP.

\subsection{Deterministic Networking on Commodity Switches}

A third, more recent line of work targets the cost barrier of TSN deployment by attempting to provide deterministic guarantees on commodity Ethernet switches. LCDN \cite{Diederich2025} is the most direct example. It deploys a centralized controller that uses Deterministic Network Calculus to perform per-flow admission control, assigns each flow to a priority queue, and routes traffic across multiple VLAN-based spanning trees on commodity strict-priority switches. Source endpoints regulate transmission with a Token Bucket Filter to enforce the arrival curves used in the admission analysis. LCDN demonstrates that competitive flow acceptance rates can be achieved on low cost switches, providing the first concrete evidence that determinism is achievable without TSN hardware.

In contrast to the works above, SbDN addresses both mixed-criticality coexistence and the cost barrier simultaneously. The TC Scheduling Agent computes collision-free release times that are enforced at the source endpoints, which removes the need for switch-side TAS hardware while preserving time-triggered determinism. The TNP method targets pure FIFO commodity switches, which neither switch-based TAS schedulers nor LCDN can support: both rely on the switch to separate traffic classes, either through scheduled gates or through strict-priority queuing. TNP avoids this requirement by separating TC and BE traffic in time at the source, so a FIFO switch never has any contention to resolve. The TP method, in turn, lets the switch resolve contention via strict-priority queuing, which simplifies the source side and allows BE traffic to be work-conserving. The BE Shaping Agent handles mixed-criticality coexistence at the source through bottleneck-aware CBS configuration, complementing rather than replacing the per-switch shaping mechanisms studied in prior mixed-criticality work.

For empirical evaluation of our proposed approaches, we select HERMES \cite{Bujosa2022} and LCDN \cite{Diederich2025} as baselines. HERMES represents the standard TSN approach for mixed-criticality scheduling on TSN-capable switches and allows us to quantify the cost of moving enforcement from switches to source endpoints. LCDN represents the closest alternative philosophy for deterministic networking on commodity switches and allows us to compare our time-triggered design against a network-calculus-based admission control approach. Both HERMES and LCDN address mixed-criticality traffic, both have publicly available open-source implementations and together they bracket the design space SbDN occupies.

Several simulation environments and testbeds support the evaluation of TSN configurations, including OMNeT++ with TSN extensions \cite{Varga2010}, hardware-based testbeds such as EnGINE \cite{Rezabek2022}, and integrated platforms that combine simulation with analytical performance models. We use INSIM \cite{Karimi2025INSIM} as our evaluation platform. INSIM provides a fully open-source modular framework with a plug-in architecture for custom schedulers and resource managers, which we use to integrate the source-based mechanisms introduced in this paper alongside the HERMES and LCDN baselines under a unified evaluation setup.

\section{System Model and Problem Definition}\label{sec:sysmodel}
This section presents the formal system model and defines the configuration problem addressed in this work. The high-level structure of the system model and problem definition draws on the formulation in~\cite{Karimi2025}, which we adapt and extend to the source-based setting introduced in this paper. We first describe the network model, followed by the workload model that captures both time-critical and best-effort communication. We then formalize the two key traffic shaping mechanisms, TAS and the CBS, whose configuration at the source endpoints forms the core of our approach. Finally, we state the scheduling and configuration problem that the proposed architecture must solve. For convenience, Table~\ref{tab:notation} summarizes the principal notations used throughout the paper.

\begin{table}[!tbp]
\centering
\caption{Summary of notations}
\label{tab:notation}
\footnotesize
\renewcommand{\arraystretch}{1.1}
\setlength{\tabcolsep}{1pt}
\begin{tabular}{@{}c c p{0.66\columnwidth}@{}}
\hline
\textbf{Symbol} & \textbf{Unit} & \textbf{Description} \\
\hline
$G = (V, E)$ & --- & Network graph \\
$V_{\mathrm{ep}}$ & --- & Set of endpoints \\
$V_{\mathrm{sw}}$ & --- & Set of intermediate switches \\
$\mathit{GCL}_v$ & --- & Gate control list at endpoint $v$ (Sec.~\ref{sec:tas_model}) \\
$idleSlope_v$ & bits/s & Vector of idle slopes at endpoint $v$ (Sec.~\ref{sec:cbs_model}) \\
$\delta_s$ & s & Processing delay at switch $s$ \\
$v^{1}_e,\; v^{2}_e$ & --- & Nodes connected by link $e$ \\
$C_e$ & bits/s & Capacity of link $e$ \\
\hline
$F$ & --- & Set of all flows \\
$v^{\mathrm{src}}_f$ & --- & Source endpoint of flow $f$ \\
$v^{\mathrm{dst}}_f$ & --- & Destination endpoint of flow $f$ \\
$pcp_f$ & --- & Traffic class of flow $f$, $pcp_f \in \mathrm{PCP}$ \\
$\mathrm{PCP}$ & --- & Set of traffic classes ($\{\mathrm{TC}, \mathrm{BE}\}$ in this work) \\
$T_f$ & s & Period of flow $f$ \\
$n_f$ & --- & Number of frames transmitted by $f$ per period\\
$d_f$ & s & End-to-end deadline of flow $f$ \\
$L_f$ & bits & Frame size of flow $f$ \\
$R_f$ & --- & Route of flow $f$ \\
$t^{\mathrm{gen}}_f$ & s & Generation time of flow $f$ within each period \\
\hline
$T_{\mathrm{cyc}}$ & s & TAS cycle duration \\
$m$ & --- & Number of segments in the GCL \\
$\Delta_i$ & s & Duration of the $i$-th GCL segment \\
$x_{i,j}$ & --- & Gate state for class $j$ in segment $i$ ($0$ or $1$) \\
\hline
$cr_{v,j}$ & bits & Credit value for class $j$ at endpoint $v$ \\
$idleSlope_{v,j}$ & bits/s & Idle slope for class $j$ at endpoint $v$ \\
$sendSlope_{v,j}$ & bits/s & Send slope for class $j$ at endpoint $v$ \\
$cr_{\max}$ & bits & Upper credit bound \\
$cr_{\min}$ & bits & Lower credit bound \\
$t_{\mathrm{wait}}$ & s & Wait time until credit recovers to zero \\
\hline
$t^{\mathrm{rel}}_f$ & s & Release time assigned to TC flow $f$ by the TC Scheduling Agent \\
$t^{\mathrm{arr}}_{f,i}$ & s & Arrival time of frame $i$ of flow $f$ at its destination \\
\hline
$F_{\mathrm{TC}}$ & --- & Set of TC flows, $\{f \in F \mid pcp_f = \mathrm{TC}\}$ \\
$T_e$ & --- & Reservation timeline on link $e$ (interval tree) \\
$\omega_f(k)$ & s & Cumulative offset from source to link $e_k$ on route $R_f$ \\
$\tau_{f,e}$ & s & Transmission time of flow $f$ on link $e$, $L_f / C_e$ \\
$\tau^{\max}_f$ & s & Transmission time of $f$ on the slowest link of $R_f$ \\
$t^{\mathrm{rel}}_{f,i}$ & s & Release time of frame $i$ of flow $f$, $t^{\mathrm{rel}}_f + i \cdot \tau^{\max}_f$ \\
$\Lambda_f$ & s & Duration of the aggregate reservation of one burst of $f$, $n_f \cdot \tau^{\max}_f$ \\
$\mathcal{B}_{e}(t_{start}, t_{end})$ & s & Forbidden release times from reservation $[t_{start}, t_{end})$ on link $e$ \\
$\mathcal{B}_f$ & s & Union of all forbidden ranges for flow $f$ \\
\hline
$F_{\mathrm{BE}}$ & --- & Set of BE flows, $\{f \in F \mid pcp_f = \mathrm{BE}\}$ \\
$S_f$ & --- & Set of intermediate switches along route $R_f$ \\
$e_{\mathrm{out}}^{s,f}$ & --- & Output link of switch $s$ along route $R_f$ \\
$N_{\mathrm{BE}}(v)$ & bits & BE bits sent from endpoint $v$ per cycle \\
$N_{\mathrm{BE}}(e)$ & bits & BE bits traversing link $e$ per cycle \\
$K_{\mathrm{BE}}(e)$ & --- & Number of BE flows whose route traverses link $e$ \\
$d_{\mathrm{in}}(e)$ & --- & Number of input links of the feeding switch carrying BE traffic to $e$ \\
$\beta_{\mathrm{BE}}(e)$ & --- & BE burst bound: worst-case number of BE frames simultaneously queued at the port feeding $e$ \\
$g_{\mathrm{be}}$ & s & Guard band duration (BE-to-TC transition) \\
$\Delta_{\mathrm{min}}$ & s & Minimum useful BE window duration \\
\hline
$\mathcal{E}_v$ & --- & Set of links traversed by any BE flow originating at $v$ \\
$b(v)$ & --- & Bottleneck link of endpoint $v$ along its BE paths \\
\hline

$\tau^{\mathrm{BE}}_e$ & s & Worst-case BE frame transmission time on link $e$, $L^{\mathrm{max}}_{\mathrm{BE}} / C_e$ \\
$L^{\mathrm{max}}_{\mathrm{BE}}$ & bits & Maximum BE frame size \\
\hline
$\omega^{B}_f(k)$ & s & Cumulative offset on route $R_f$ including BE blocking \\
\hline
\end{tabular}
\end{table}

\subsection{Network Model}

We model the network as an undirected graph $G = (V, E)$, where $V$ denotes the set of network devices and $E$ denotes the set of full-duplex communication links. The node set is partitioned as $V = V_{\mathrm{ep}} \cup V_{\mathrm{sw}}$, where $V_{\mathrm{ep}}$ is the set of endpoints and $V_{\mathrm{sw}}$ is the set of intermediate switches.

Each endpoint $v \in V_{\mathrm{ep}}$ is defined as
\begin{equation}
v = (\mathit{GCL}_v,\; idleSlope_v),
\end{equation}

where $\mathit{GCL}_v$ is the gate control list and $idleSlope_v$ is the vector of idle slopes assigned to the endpoint. These parameters govern the traffic shaping behavior (TAS and CBS) at the source and are defined in detail in Sections~\ref{sec:tas_model} and~\ref{sec:cbs_model}, respectively.

Each switch $s \in V_{\mathrm{sw}}$ comes with a parameter $\delta_s$ denoting the processing delay per frame.

Each link $e \in E$ is defined as
\begin{equation}
e = (v^{1}_e,\; v^{2}_e,\; C_e),
\end{equation}

where $v^{1}_e \in V$ and $v^{2}_e \in V$ are the two nodes connected by the link and $C_e$ is the link capacity in bits per second, available in each direction simultaneously. We assume at most one link between any pair of nodes. Signal propagation delays on the links are neglected, as they are orders of magnitude smaller than frame transmission times and switch processing delays in the short-range networks targeted in this work; if needed, they can be incorporated as an additive per-link constant without affecting the structure of the proposed methods.

\subsection{Workload Model}\label{sec:workload-model}

The workload consists of a set of flows $F$ carried over the network. Each flow $f \in F$ is defined as
\begin{equation}
f = (v^{\mathrm{src}}_f,\; v^{\mathrm{dst}}_f,\; pcp_f,\; T_f,\; n_f,\; d_f,\; L_f,\; R_f,\; t^{\mathrm{gen}}_f),
\end{equation}

where $v^{\mathrm{src}}_f \in V_{\mathrm{ep}}$ is the source endpoint, $v^{\mathrm{dst}}_f \in V_{\mathrm{ep}}$ is the destination endpoint, $pcp_f \in \mathrm{PCP}$ is the traffic class of the flow, $T_f$ is the transmission period, $d_f$ is the end-to-end deadline with $d_f \leq T_f$, $n_f$ is the number of frames transmitted per period, $L_f$ is the frame size in bits, $R_f = (e_1, e_2, \ldots)$ is the route of the flow defined as the ordered sequence of links traversed from source to destination, and $t^{\mathrm{gen}}_f$ is the generation time at which the application produces its first frame within each period. Each flow periodically generates one or more frames per period. If more than one frame is generated, this occurs in a burst of consecutive frames that are generated from an application-level data payload. These frames must traverse the flow's route and arrive at the destination within $d_f$ time units of the generation of the first frame.

The traffic class $pcp_f$ determines the priority and scheduling treatment of a flow. The TSN standard supports up to eight distinct traffic classes. In this work, we consider two: time-critical (TC) and best-effort (BE); that is, $\mathrm{PCP} = \{\mathrm{TC},\; \mathrm{BE}\}$, with TC having strictly higher priority. TC flows carry delay-sensitive data and must meet their deadlines. BE flows do not have hard deadline constraints. For a BE flow, the parameters $T_f$ and $d_f$ are not used. Restricting the methods to two classes is not a limitation. Intermediate priority classes typically used in TSN, such as AVB, exist to let switches differentiate between flows with bounded but non-strict latency requirements. In SbDN, the TC Scheduling Agent assigns each TC flow its own collision-free release time at the source. Thus, flows with bounded latency requirements such as AVB can be scheduled as TC, and others are treated as BE.

\subsection{Source-Based TAS Model}
\label{sec:tas_model}

TAS, defined by the IEEE 802.1Qbv standard, is a traffic shaping mechanism that provides temporal isolation between traffic classes. TAS operates by partitioning time into a repeating cycle of duration $T_{\mathrm{cyc}}$. The cycle duration is set to the hyperperiod of the TC flows, that is, the least common multiple of their periods $T_f$, so that the placement of all periodic TC frames repeats identically in every cycle. Within each cycle, a Gate Control List (GCL) specifies a sequence of $m$ non-overlapping time segments $\Delta_1, \Delta_2, \ldots, \Delta_m$. Each segment determines which traffic classes are permitted to transmit through a binary gate state per class, as shown in Table~\ref{tab:gcl_structure}.

\begin{table}[h]
\centering
\caption{Structure of a GCL}
\label{tab:gcl_structure}
\begin{tabular}{ccccc}
\hline
\textbf{Segment} & $pcp_1$ & $pcp_2$ & $\cdots$ & $pcp_{|\mathrm{PCP}|}$ \\
\hline
$\Delta_1$ & $x_{1,1}$ & $x_{1,2}$ & $\cdots$ & $x_{1,|\mathrm{PCP}|}$ \\
$\Delta_2$ & $x_{2,1}$ & $x_{2,2}$ & $\cdots$ & $x_{2,|\mathrm{PCP}|}$ \\
$\vdots$   & $\vdots$  & $\vdots$  & $\ddots$ & $\vdots$ \\
$\Delta_m$ & $x_{m,1}$ & $x_{m,2}$ & $\cdots$ & $x_{m,|\mathrm{PCP}|}$ \\
\hline
\end{tabular}
\end{table}

Here, $x_{i,j} \in \{0, 1\}$ indicates whether the gate for traffic class $j$ is open ($x_{i,j} = 1$) or closed ($x_{i,j} = 0$) during segment $\Delta_i$. Multiple traffic classes may be open simultaneously within the same segment.

A segment in which all gates are closed is called a guard band. Guard bands are inserted at the boundary between segments of different traffic classes to account for frames that are already in transit within the network. Because switches operate in store-and-forward mode and do not preempt frames, a frame that has begun transmission on a link will complete before yielding the medium. The guard band provides sufficient time for all such in-flight frames to be fully delivered before the next traffic class begins transmission, thereby preventing interference between classes.

In the standard TSN deployment, a GCL is defined per egress port at every network device, including all intermediate switches. In this work, TAS is configured and enforced exclusively at the source endpoints. Each endpoint has a single egress port, so a single $\mathit{GCL}_v$ is assigned per (source) endpoint $v$. A GCL is therefore a per-endpoint object; how its contents are computed, and whether the same GCL is shared across endpoints, depends on the scheduling method and is detailed in Section~\ref{sec:tas_agent}. Throughout the remainder of this paper, TAS refers to this source-based variant unless stated otherwise. Furthermore, since we consider two traffic classes ($\mathrm{PCP} = \{\mathrm{TC},\; \mathrm{BE}\}$), the gate state per segment reduces to two binary values $(x_{i,\mathrm{TC}},\; x_{i,\mathrm{BE}})$.

\subsection{Source-Based CBS Model}
\label{sec:cbs_model}

CBS, defined by the IEEE 802.1Qav standard, is a traffic shaping mechanism that regulates the transmission rate of a traffic class by maintaining a credit counter. The credit determines whether a frame is eligible for transmission: a frame may only be transmitted when its associated credit is non-negative.

For each traffic class $j \in \mathrm{PCP}$ at an endpoint $v$, the CBS maintains a credit value $cr_{v,j}$ that evolves over time according to two parameters: the idle slope $idleSlope_{v,j}$ and the send slope $sendSlope_{v,j}$. When no frame of class $j$ is being transmitted, the credit accumulates at rate $idleSlope_{v,j} > 0$. During transmission of a frame of class $j$, the credit decreases at rate $sendSlope_{v,j} \le 0$.\footnote{A 0 send slope means that no actual shaping is applied.} When the transmission queue is empty, the credit cannot accumulate beyond zero; it recovers toward zero at the idle slope rate but is capped at zero until a new frame arrives. The credit is bounded by
\begin{equation}
cr_{\min} \leq cr_{v,j} \leq cr_{\max},
\end{equation}
where $cr_{\max}\ge 0$ and $cr_{\min} < 0$ are the upper and lower credit bounds, respectively. Credit can go above zero when a frame is delayed by higher-priority traffic, as the idle slope continues to accumulate credit during the waiting period. The upper bound $cr_{\max}$ limits this accumulation to prevent the endpoint from transmitting a burst of back-to-back frames once the medium becomes available. A frame of class $j$ at endpoint $v$ is eligible for transmission only when $cr_{v,j} \geq 0$. If credit is negative, the endpoint must wait for a recovery time of
\begin{equation}
t_{\mathrm{wait}} = \frac{|cr_{v,j}|}{idleSlope_{v,j}}
\end{equation}
before the next frame becomes eligible.

The idle slope governs the long-term average transmission rate of the traffic class and is the primary configuration parameter of the CBS. The send slope is derived from the idle slope and the link capacity $C_e$ of the endpoint's egress link as
\begin{equation}
sendSlope_{v,j} = idleSlope_{v,j} - C_e.
\end{equation}

The vector of idle slopes at endpoint $v$, denoted $idleSlope_v = [idleSlope_{v,1}, \ldots, idleSlope_{v,|PCP|}]$, constitutes the CBS configuration introduced in the network model. The choice of idle slope values determines both the per-class throughput guarantee and the worst-case queue buildup at downstream switches. The mechanism can in principle be configured for any traffic class; the specific subset of classes that is shaped depends on the scheduling method, as detailed in Section~\ref{sec:SbDN}.

In the standard TSN deployment, CBS is defined per egress port at every switch in the network. In this work, as with TAS, CBS is configured and enforced at the source endpoints only. Throughout the remainder of this paper, CBS refers to this source-based variant unless stated otherwise.

\subsection{Problem Definition}
\label{sec:problem_definition}

Given a network $G = (V, E)$ and a set of flows $F$ as defined above, the primary objective is to configure the network such that every TC flow $f$ with $pcp_f = \mathrm{TC}$ meets its end-to-end deadline. Let $t^{\mathrm{arr}}_{f,i}$ denote the arrival time of a frame $i$ of flow $f$ at its destination. The deadline constraint requires:
\begin{equation}
t^{\mathrm{arr}}_{f,i} - t^{\mathrm{gen}}_f \leq d_f, \quad \text{for~all~} f \in F
\;\text{with}\; pcp_f = \mathrm{TC}.
\end{equation}

The mechanisms available to achieve this requirement are: (i) release-time assignment, which computes for each TC flow a release time $t^{\mathrm{rel}}_f \geq t^{\mathrm{gen}}_f$ at which the source endpoint is enforced to begin transmission, releasing the $n_f$ frames of each period consecutively from this time onward, with the constraint that all frames of the flow must reach their destination by $t^{\mathrm{gen}}_f + d_f$; (ii) TAS, which configures $\mathit{GCL}_v$ at each endpoint to partition time into segments where specific traffic classes are permitted to transmit; (iii) CBS, which configures $idleSlope_v$ (a single value per endpoint in this work, since we consider one BE class) at each endpoint to regulate the transmission rate of BE traffic through CBS; and (iv) strict-priority queuing.

The problem is to determine which combination of these mechanisms to employ, and with what configuration of their parameters ($t^{\mathrm{rel}}_f$, $\mathit{GCL}_v$, $idleSlope_v$), such that the deadline constraint is satisfied for every TC flow. Different combinations lead to different trade-offs between switch complexity, scheduling complexity, and the strength of the guarantees provided. In Section~\ref{sec:SbDN}, we present two architectural instantiations that address this problem using different subsets of the available mechanisms.

\section{SbDN Architecture}\label{sec:SbDN}
In this section, we propose two source-based architectural scenarios for deterministic communication, each employing a different combination of the mechanisms mentioned above. The architecture consists of three principal components: a central network controller that computes all scheduling configurations, source endpoints that enforce them locally, and commodity switches that serve as pure forwarding elements.

\subsection{Overview of the Methods}\label{sec:SbDN-overview}

All scheduling intelligence resides in the central network controller, which is composed of three cooperating agents. The TC Scheduling Agent computes the release time $t^{\mathrm{rel}}_f$ for each TC flow. The Network Partitioning Agent computes the gate control list $\mathit{GCL}_v$ for each endpoint. The BE Shaping Agent computes the idle slope vector $idleSlope_v$ for each endpoint. Once computed, these configurations are pushed to the endpoints, which enforce them locally: release time enforcement ensures that TC frames are transmitted precisely at their scheduled instants, source-based TAS gates traffic according to $\mathit{GCL}_v$, and source-based CBS regulates the transmission rate according to $idleSlope_v$. The switches remain unaware of these configurations and perform no scheduling. They operate as commodity forwarding elements with per-port queues, optionally supporting strict-priority queuing.

We assume that devices in the network are time-syn\-chronized. This can be achieved through hardware-based synchronization, where switches implement the Precision Time Protocol (PTP) \cite{linuxptp}, or through software-based solutions \cite{V-TSN}. The choice between the two depends on the desired accuracy. Importantly, hardware-based synchronization does not require full TSN-capable switches; it can be provided by commodity switches that support PTP, such as the FS IES3100-8TF \cite{FSIES3100}. This keeps the overall network cost low while preserving the timing accuracy needed for deterministic operation.

We present two methods that differ in how the available mechanisms are combined. Temporal Network Partitioning (TNP) relies on pure FIFO commodity switches and moves all scheduling responsibility to the source endpoints, combining release-time enforcement, source-based TAS, and source-based CBS. Traffic Prioritization (TP) relies on commodity switches with strict-priority queuing and simplifies the source side by combining release-time enforcement with source-based CBS only, removing the need for a source GCL. Table~\ref{tab:scenarios} summarizes the mechanisms employed by each scenario and the design trade-offs they present.

\begin{table}[h]
\centering
\caption{Mechanisms and trade-offs of the proposed scenarios}
\label{tab:scenarios}
\begin{tabular}{lcc}
\hline
\textbf{Aspect} & \textbf{TNP} & \textbf{TP} \\
\hline
Central scheduler agents & TC, NP, BE & TC, BE \\
Source enforcement & TC, NP, BE & TC, BE \\
Switch capability & FIFO only & Strict priority \\
Switch complexity & Minimum & Moderate \\
Scheduling complexity & High & Lower \\
\hline
\end{tabular}
\end{table}

Strict-priority queuing is a widely available feature in commodity Ethernet switches and does not require TSN support. Common enterprise-grade switches such as the Cisco Catalyst 2960 series \cite{Cisco2960} support strict-priority queuing, enabling TP on existing non-TSN infrastructure. When hardware-based time synchronization is also required, switches such as the FS IES3100-8TF provide both strict-priority queuing and PTP in a single cost-efficient platform.

\subsection{Temporal Network Partitioning}
\label{sec:scenario_a}

TNP supports networks with pure FIFO commodity switches; all scheduling responsibility is moved to the source endpoints. The central controller employs all three agents introduced earlier: the TC Scheduling Agent, the Network Partitioning Agent, and the BE Shaping Agent. The three agents operate in a pipeline, each building on the output of the previous one. The TC Scheduling Agent first assigns collision-free release times to all TC frames, as illustrated for the example workload in Fig.~\ref{fig:illust_tnp_perlink}. The Network Partitioning Agent then uses these reservations to synthesize a global gate control list $\mathit{GCL}$ enforced at each endpoint; the resulting cycle structure for the example is shown in Fig.~\ref{fig:illust_tnp_gcl}. Finally, the BE Shaping Agent computes the idle slope vector $idleSlope_v$ for each endpoint. The resulting configuration is pushed to the endpoints, which enforce it locally: TC frames are transmitted at their scheduled release times, BE transmission is gated by the source GCL, and the BE transmission rate is regulated by the CBS credit model.

Within each TAS cycle, the schedule contains multiple TC and BE phases that alternate based on the pattern of TC reservations produced by the TC Scheduling Agent, as in the cycle shown in Fig.~\ref{fig:illust_tnp_gcl}. A distinctive feature of TNP is that it is \textit{non-work-conserving} for BE traffic: BE frames cannot be transmitted during TC phases even if the medium is momentarily idle. This property is a deliberate trade-off. By enforcing strict temporal separation between TC and BE traffic, the architecture eliminates the possibility of BE frames interfering with TC frames at any switch, thereby guaranteeing deterministic TC latency on FIFO switches without requiring strict-priority queuing.

\subsubsection{TC Scheduling Agent}
\label{sec:release_time_agent}

The TC Scheduling Agent determines the release time $t^{\mathrm{rel}}_f$ of each TC flow $f \in F_{\mathrm{TC}}$, where $F_{\mathrm{TC}} = \{f \in F \mid pcp_f = \mathrm{TC}\}$. Since the frames within one flow within one transmission period are assumed to arrive in a burst (Section \ref{sec:workload-model}), the TC Scheduling Agent releases those frames consecutively, spaced by the transmission time of the flow on the slowest link of its route; this is the smallest spacing that guarantees that frames never queue behind each other anywhere on the route. Hence, it suffices to determine a release time for a flow and then reserve a sufficiently large time segment on every link for the transmission of all its frames. The objective is to ensure that no two TC frames occupy the same link at the same time.

Because different flows traverse different routes, a collision-free assignment must be verified on every link individually. The agent maintains a per-link reservation timeline $T_e$ for each link $e \in E$, implemented as an interval tree. Each interval $[t_{start}, t_{end}) \in T_e$ represents a time window during which a specific TC frame occupies link $e$, where $t_{start}$ and $t_{end}$ denote the start and end times of the reservation, respectively. A candidate release time $t^{\mathrm{rel}}_f$ is admissible if, and only if, for every link $e$ on the route $R_f$, the resulting reservation interval does not overlap with any existing interval in $T_e$.

The agent processes TC flows one at a time following Earliest Deadline First (EDF) \cite{EDF} ordering: flows with tighter deadlines are scheduled first, since they have less flexibility in their placement. For each flow, the agent runs a \emph{forbidden-interval scan} (Algorithm~\ref{alg:rt_agent}) that finds the earliest admissible release time along the entire route.

For a flow $f$ with route $R_f = (e_1, e_2, \ldots, e_{|R_f|})$,  let
\begin{equation}\label{eq:taumax}
\tau^{\max}_f = \max_{e \in R_f} \tau_{f,e} = \frac{L_f}{\min_{e \in R_f} C_e}
\end{equation}
denote the transmission time of the flow on the slowest link of its route, with $\tau_{f,e} = L_f / C_e$ the transmission time of one frame on link $e$.\footnote{The frame size $L_f$ is taken to include the mandatory 96-bit Ethernet inter-frame gap, so $\tau_{f,e}=L_f/C_e$ accounts for the required minimum spacing between consecutive frames.} The agent releases the $n_f$ frames of one transmission period spaced by this slowest-link transmission time: the $i$-th frame, $i \in \{0, \ldots, n_f - 1\}$, is released at
\begin{equation}\label{eq:frame_release}
t^{\mathrm{rel}}_{f,i} = t^{\mathrm{rel}}_{f} + i \cdot \tau^{\max}_f.
\end{equation}
This spacing guarantees that no frame ever queues behind an earlier frame of its own flow: frame $i$ becomes available at each link of the route exactly $\tau^{\max}_f$ after frame $i-1$, while frame $i-1$ occupies link $e_k$ for only $\tau_{f,e_k} \leq \tau^{\max}_f$. Every frame therefore traverses the route without queueing, and frame $i$ arrives at the beginning of link $e_k$ at time $t^{\mathrm{rel}}_{f,i} + \omega_f(k)$, where
\begin{equation}\label{eq:offset}%
\omega_f(k) = \sum_{j=1}^{k-1} \left( \tau_{f,e_j} + \delta_{s_j} \right),
\end{equation}
with $\delta_{s_j}$ the processing delay at the switch that forwards the frame from link $e_j$ to link $e_{j+1}$. Frame $i$ occupies link $e_k$ during the interval $[t^{\mathrm{rel}}_{f,i} + \omega_f(k),\; t^{\mathrm{rel}}_{f,i} + \omega_f(k) + \tau_{f,e_k})$.
Rather than tracking each frame individually, the agent represents the $n_f$ frames of one transmission period as a single aggregate reservation of duration
\begin{equation}\label{eq:lambda}
\Lambda_f = n_f \cdot \tau^{\max}_f
\end{equation}
on every link of the route. Since the first frame starts transmission on link $e_k$ at $t^{\mathrm{rel}}_{f} + \omega_f(k)$ and the last frame completes there at $t^{\mathrm{rel}}_{f} + (n_f - 1)\,\tau^{\max}_f + \omega_f(k) + \tau_{f,e_k} \leq t^{\mathrm{rel}}_{f} + \omega_f(k) + \Lambda_f$, this equal-size interval covers the transmission of the entire burst on every link. This keeps the per-link interval trees compact. The reservation is exact on the slowest link of the route and, in particular, on every link when all link capacities are equal, as in our examples and evaluation; on faster links, it is slightly conservative.

For any reservation $[t_{start}, t_{end}) \in T_{e_k}$ on link $e_k$, the set of release times $t^{\mathrm{rel}}_{f}$ for which the aggregate reservation of $f$ on $e_k$ would overlap with $[t_{start}, t_{end})$ forms the forbidden range
\begin{align}
\mathcal{B}_{e_k}(&t_{start}, t_{end}) = \\ \nonumber &\left[ t_{start} - \omega_f(k) - \Lambda_f,\; t_{end} - \omega_f(k) \right).
\end{align}

The union of all forbidden ranges across all reservations on links in $R_f$, denoted $\mathcal{B}_f$, defines the set of release times that would cause a collision for flow $f$. Since the cycle duration $T_{\mathrm{cyc}}$ is the hyperperiod of the TC flows, flow $f$ transmits its burst once per transmission period and therefore $T_{\mathrm{cyc}}/T_f$ times per cycle. The $q$-th instance, $q \in \{0, \ldots, T_{\mathrm{cyc}}/T_f - 1\}$, is released at $t^{\mathrm{rel}}_f + q \cdot T_f$ and hence occupies every link of the route $q \cdot T_f$ later than instance $0$. A reservation $[t_{start}, t_{end})$ on a link therefore contributes one forbidden range per instance, namely $\mathcal{B}_{e_k}(t_{start}, t_{end})$ shifted by $-q \cdot T_f$: if the base release time $t^{\mathrm{rel}}_f$ were in this shifted range, then instance $q$, released $q \cdot T_f$ later, would fall exactly into the reservation. The scan, Algorithm~\ref{alg:rt_agent}, accumulates these ranges over all reservations and all instances, keeping $\mathcal{B}_f$ sorted and merged (lines~4--6). It then selects the earliest admissible release time (line~7): a flow is admitted only if there is a release time $t^{\mathrm{rel}}_f \geq t^{\mathrm{gen}}_f$ with $t^{\mathrm{rel}}_f \notin \mathcal{B}_f$ at which even the last frame of the burst reaches the destination within the deadline; that is,
\begin{equation}
t^{\mathrm{rel}}_f + (n_f - 1)\,\tau^{\max}_f + \omega_f(|R_f|) + \tau_{f,e_{|R_f|}} \leq t^{\mathrm{gen}}_f + d_f.
\label{eq:admission_tnp}
\end{equation}
If no such release time exists, the flow is declared infeasible (lines~8--9). Once a valid $t^{\mathrm{rel}}_f$ is found, the agent inserts the aggregate reservations of all $T_{\mathrm{cyc}}/T_f$ instances into the timeline of every link in $R_f$ (lines~10--12). Subsequent flows are scheduled against the updated timelines, and the outer loop (line~2) repeats until all TC flows have been processed.

\begin{algorithm}[t]
\caption{Release-Time Assignment}
\label{alg:rt_agent}
\begin{algorithmic}[1]
\Require Set of TC flows $F_{\mathrm{TC}}$, network $G = (V, E)$, cycle duration $T_{\mathrm{cyc}}$
\Ensure Release time $t^{\mathrm{rel}}_f$ for every admitted flow $f \in F_{\mathrm{TC}}$ and reservation timelines $T_e$ for every $e \in E$
\State $T_e \gets \emptyset$ for every $e \in E$;\quad sort $F_{\mathrm{TC}}$ in EDF order
\For{each flow $f \in F_{\mathrm{TC}}$}
    \State $\tau^{\max}_f \gets \max_{e \in R_f} \tau_{f,e}$;\quad $\Lambda_f \gets n_f \cdot \tau^{\max}_f$;\quad $\mathcal{B}_f \gets \emptyset$
    \For{each link $e_k \in R_f$}
        \For{$[t_{start}, t_{end}) \in T_{e_k}$,  $q \in \{0, \ldots, \frac{T_{\mathrm{cyc}}}{T_f} - 1\}$}
            \State \begin{tabular}[t]{@{\hspace{-3pt}}l}Insert $\mathcal{B}_{e_k}(t_{start}, t_{end}) - q \cdot T_f$ into $\mathcal{B}_f$, keep-\\
            ing $\mathcal{B}_f$ sorted and merging overlapping ranges\end{tabular}
        \EndFor
    \EndFor
    \State $t^{\mathrm{rel}}_f \gets$ earliest $t \geq t^{\mathrm{gen}}_f$ with $t \notin \mathcal{B}_f$ satisfying Eq.~\eqref{eq:admission_tnp}
    \If{no such $t$ exists}
        \State Declare $f$ infeasible; \textbf{continue} with the next flow
    \EndIf
    \For{each $e_k \in R_f$, $q \in \{0, \ldots, \frac{T_{\mathrm{cyc}}}{T_f} - 1\}$}
        \State $start_{q,k} \gets t^{\mathrm{rel}}_f + q \cdot T_f + \omega_f(k)$
        \State $T_{e_k} \gets T_{e_k} \cup \{[\,start_{q,k},\; start_{q,k} + \Lambda_f\,)\}$  
    \EndFor
\EndFor
\end{algorithmic}
\end{algorithm}

After all TC flows have been processed, the per-link timelines jointly define a collision-free schedule for the admitted flows: for every link $e \in E$, the intervals in $T_e$ are pairwise disjoint. Consequently, at most one admitted TC frame occupies any link at any time, and no two TC frames can arrive at the same switch output port simultaneously. This eliminates contention between TC streams at every switch and guarantees that TC frames do not queue behind one another anywhere in the network. Handling infeasible flows, for example by revising the application quality, is the task of a supervisory controller and out of the scope of this work.

\paragraph*{Complexity}
Regarding computational complexity, the cost of the agent is governed by the number of reservations, which depends on the hyperperiod: each admitted flow $f$ contributes $|R_f| \cdot T_{\mathrm{cyc}}/T_f$ aggregate base reservations, giving a total of $r = \sum_{f \in F_{\mathrm{TC}}} |R_f| \cdot T_{\mathrm{cyc}}/T_f$ base reservations across all timelines if all flows are admitted. In addition, the scan creates derived reservations, the forbidden ranges of $\mathcal{B}_f$: when flow $f$ is placed, every base reservation on the links of its route contributes one range for each of the $T_{\mathrm{cyc}}/T_f$ transmission instances of $f$ in a hyperperiod. Hence, at most $r \cdot T_{\mathrm{cyc}}/T_f$ derived ranges are added per flow. So, the total number of derived reservations over all flows that are added is at most $r \cdot \sum_{f \in F_{\mathrm{TC}}} T_{\mathrm{cyc}}/T_f \leq r^{2}$, a bound that is essentially attained when all flows share a single link.

Using balanced binary search trees (BSTs) for both the $\mathcal{B}_f$ and $T_e$, the loop of line 2 essentially inserts at most $r + r^2$ reservations in BSTs of size at most $r + r^2$, which has worst-case complexity $O\big((r+r^2)\log(r+r^2)\big) = O(r^{2}\log r)$; with low contention, the number of derived ranges added to $\mathcal{B}_f$ approaches $r$ and the cost approaches $O(r \log r)$. 

\subsubsection{Network Partitioning Agent}
\label{sec:tas_agent}

The Network Partitioning Agent uses the per-link reservation timelines $T_e$ produced by the TC Scheduling Agent to synthesize a gate control list. These timelines admit two constructions. The first projects all TC reservations onto a single time axis and enforces the resulting GCL identically at every endpoint; we refer to this as the \emph{global} GCL. The second synthesizes a more compact $\mathit{GCL}_v$ per endpoint, retaining only the TC segments in which that endpoint transmits and closing the remainder. The two are not equivalent in general: because a BE frame emitted by an endpoint may traverse several links before reaching its destination, an endpoint must keep its BE gate closed during any TC phase whose frames share a link with that endpoint's BE routes, not merely during TC activity on its own egress link. In topologies dominated by a shared bottleneck, this condition forces every endpoint to close BE during every TC phase, so the per-endpoint GCLs differ from the global GCL only in their number of segments, not in the time available for BE; with greater path diversity, per-endpoint GCLs can keep BE open during TC phases that do not touch an endpoint's BE routes, improving BE throughput at the cost of a per-link drain analysis and per-endpoint synthesis. The current version of SbDN adopts the global GCL: enforcing one identical GCL at every endpoint guarantees that BE is closed network-wide whenever any TC frame is in flight, which admits the single network-wide drain argument of Section~\ref{sec:scenario_a_proof}. This choice requires all endpoints to share a common notion of time, since any misalignment could allow BE traffic from one endpoint to leak into a TC phase at a shared downstream link.

Within each cycle of duration $T_{\mathrm{cyc}}$, the schedule alternates between TC phases and BE phases based on where TC reservations appear on the timelines. For transition from a TC phase to a BE phase, no guard band is required, because the TC Scheduling Agent reserves time on every link of every TC flow's route. When the last TC reservation of a TC phase ends, no TC frame is left in flight anywhere in the network. BE transmission can therefore begin immediately.

In the transition from a BE phase to a TC phase, however, a guard band is required. During the BE phase, BE frames may be in flight at any switch in the network. Before the next TC phase begins, all in-flight BE frames must reach their destinations to avoid interfering with TC traffic at any downstream switch. We denote the required guard band duration by $g_{\mathrm{be}}$.

The goal of $g_{\mathrm{be}}$ is to guarantee that the network is fully drained of BE traffic before any TC phase starts. This requires bounding both the transmission time of BE frames along their routes and the worst-case queue buildup they may experience at intermediate switches. The formula for $g_{\mathrm{be}}$ is therefore coupled to the CBS configuration, which determines the maximum BE queue depth at each switch. The derivation of the CBS idle slope is presented in Section~\ref{sec:cbs_agent}, and the formal proof that the combination of $g_{\mathrm{be}}$ and the chosen idle slope guarantees a fully drained network is presented in Section~\ref{sec:scenario_a_proof}. The guard band is given by
\begin{equation}
\begin{split}
g_{\mathrm{be}} = \max_{f \in F_{\mathrm{BE}}} \Bigg[ 
& \sum_{e \in R_f} \tau_{f,e} \;+ \\
& \sum_{s \in S_f} \Big( \delta_s + (\beta_{\mathrm{BE}}(e_{\mathrm{out}}^{s,f}) - 1) \cdot \tau^{\mathrm{BE}}_{e_{\mathrm{out}}^{s,f}} \Big) \Bigg].
\end{split}
\label{eq:guard_be}
\end{equation}

Intuitively, $g_{\mathrm{be}}$ upper-bounds the time the worst-case BE frame needs to leave the network. The first sum is the frame's own transmission along every link of its route. The second sum adds, at each switch on the way, the processing delay $\delta_s$ and the time to clear the BE frames that may already sit ahead of it at that output port. $F_{\mathrm{BE}} = \{f \in F \mid pcp_f = \mathrm{BE}\}$ is the set of BE flows, $S_f$ the intermediate switches along $R_f$, $e_{\mathrm{out}}^{s,f}$ the output link of switch $s$ on the path of $f$, and $\tau^{\mathrm{BE}}_e = L^{\mathrm{max}}_{\mathrm{BE}}/C_e$ the transmission time of a maximum-size BE frame on $e$.

$\beta_{\mathrm{BE}}(e)$ is the \emph{BE burst bound} for link $e$: the worst-case number of BE frames that can simultaneously arrive at the port feeding $e$.
This number is bounded by both the number of BE flows and the number of incoming links. Hence,
\begin{equation}
\beta_{\mathrm{BE}}(e) = \min\big(K_{\mathrm{BE}}(e),\, d_{\mathrm{in}}(e)\big),
\label{eq:beta_be}
\end{equation}
where $K_{\mathrm{BE}}(e) = |\{f \in F_{\mathrm{BE}} \mid e \in R_f\}|$ is the number of BE flows traversing $e$ and $d_{\mathrm{in}}(e)$ is the number of input links of the switch feeding $e$. Together with the BE send rates resulting from the traffic shaping (defined below in \ref{sec:cbs_agent}) this provides a bound on the queue depth (as shown in \ref{sec:scenario_a_proof}).

The per-switch term in \eqref{eq:guard_be} adds the processing delay $\delta_s$ and the queuing delay sequentially. In practice, these can partly overlap, since a switch may receive and process a frame while its output link is still transmitting the frames queued ahead of it. The exact overlap depends on the internal pipeline of the switch. The hardware-independent model used in this version of SbDN takes the conservative sum, which is a valid upper bound for any commodity switch. When the pipeline of the target switch is known, this term can be replaced by a tighter switch-specific bound.

Given the guard band, the Network Partitioning Agent constructs the GCL in two passes, given in full in Algorithm~\ref{alg:tas_agent}. The first pass projects all TC reservations from every timeline onto a single time axis (line~1), sorting them by start time, and groups reservations into clusters (lines~2--10): two reservations are merged into the same cluster when the gap between them is smaller than $g_{\mathrm{be}} + \Delta_{\mathrm{min}}$ (lines~5--6), where $\Delta_{\mathrm{min}}$ is the minimum useful BE window duration; otherwise the current cluster is closed and a new one is started (lines~7--9). Note that the reservations of one flow on consecutive links of its route are separated by at most the processing delay of the switch between them; choosing $\Delta_{\mathrm{min}} \geq \max_{s \in V_{\mathrm{sw}}} \delta_s$ therefore guarantees that these gaps are always merged into a single TC phase, so no BE phase can open while a TC frame resides inside a switch. The second pass initializes an empty GCL (lines~11--12) and walks the clusters in order (line~13): wherever the gap before a cluster is large enough, it inserts a BE phase followed by a guard band of duration $g_{\mathrm{be}}$ (lines~14--16) and then appends the cluster as a TC segment (line~17). If a final BE phase and guard band fit, these are appended to fill any remaining gap to the end of the cycle (lines~19--21); otherwise, the GCL is completed with a guard band of the remaining time frame only (lines~22--23).

\begin{algorithm}[h]
\caption{GCL Synthesis}
\label{alg:tas_agent}
\begin{algorithmic}[1]
\Require Per-link timelines $\{T_e\}_{e \in E}$, cycle duration $T_{\mathrm{cyc}}$, guard band $g_{\mathrm{be}}$, minimum BE window $\Delta_{\mathrm{min}}$
\Ensure Global gate control list $\mathit{GCL}$
\State \begin{tabular}[t]{@{}l}$\mathcal{I} \gets$ union of all intervals across $\{T_e\}_{e \in E}$, \\ sorted by start time\end{tabular}
\State $\mathcal{C} \gets \emptyset$ \Comment{set of TC clusters}
\State $[t_{start}, t_{end}) \gets$ first interval in $\mathcal{I}$
\For{each subsequent interval $[t'_{start}, t'_{end}) \in \mathcal{I}$}
    \If{$t'_{start} - t_{end} < g_{\mathrm{be}} + \Delta_{\mathrm{min}}$}
        \State $t_{end} \gets \max(t_{end}, t'_{end})$ \Comment{\begin{tabular}[t]{@{}l}merge into current \\ cluster\end{tabular}}
    \Else
        \State $\mathcal{C} \gets \mathcal{C} \cup \{[t_{start}, t_{end})\}$
        \State $[t_{start}, t_{end}) \gets [t'_{start}, t'_{end})$
    \EndIf
\EndFor
\State $\mathcal{C} \gets \mathcal{C} \cup \{[t_{start}, t_{end})\}$
\State Initialize $GCL \gets$ empty list of segments
\State $t_{\mathrm{cur}} \gets 0$
\For{each cluster $[t_{start}, t_{end}) \in \mathcal{C}$}
    \If{$t_{\mathrm{cur}} < t_{start}$}
        \State \begin{tabular}[t]{@{}l}Append BE segment of duration \\ \qquad\qquad $t_{start} - t_{\mathrm{cur}} - g_{\mathrm{be}}$ to $\mathit{GCL}$\end{tabular}
        \State Append guard band of duration $g_{\mathrm{be}}$ to $\mathit{GCL}$
    \EndIf
    \State Append TC segment of duration $t_{end} - t_{start}$ to $\mathit{GCL}$
    \State $t_{\mathrm{cur}} \gets t_{end}$
\EndFor
\If{$t_{\mathrm{cur}} < T_{\mathrm{cyc}}-(g_{\mathrm{be}} + \Delta_{\mathrm{min}})$}
    \State \begin{tabular}[t]{@{}l}Append BE segment of duration \\ \qquad\qquad  $T_{\mathrm{cyc}} - t_{\mathrm{cur}} - g_{\mathrm{be}}$ to $\mathit{GCL}$\end{tabular}
    \State Append guard band of duration $g_{\mathrm{be}}$ to $\mathit{GCL}$
\Else
    \State \begin{tabular}[t]{@{}l}Append guard band of duration $T_{\mathrm{cyc}} - t_{\mathrm{cur}}$\\ \qquad\qquad  till end of cycle to $\mathit{GCL}$\end{tabular}
\EndIf
\end{algorithmic}
\end{algorithm}
The resulting GCL is pushed identically to every endpoint. Because all endpoints share the same cycle structure and TC reservations are already collision-free, the GCL enforces temporal separation between TC and BE traffic throughout the entire network while preserving the release times computed by the TC Scheduling Agent. 

\paragraph*{Complexity} The complexity of Algorithm~\ref{alg:tas_agent} is determined by the merge of the $r$ TC reservations collected from all timelines (line~1), the clustering pass, and the segment-construction pass. All three are linear in $r$, assuming the $T_e$ are sorted to start with. As explained for the TC Scheduling Agent, this cost $r$ depends on the hyperperiod of all the flows.

\subsubsection{BE Shaping Agent}
\label{sec:cbs_agent}
The BE Shaping Agent computes the idle slope vector $idleSlope_v$ enforced at each endpoint. In TNP, the idle slope for TC traffic is irrelevant because TC frames are released at fixed times determined by the TC Scheduling Agent and gated by the TC segments of the GCL. The BE Shaping Agent therefore focuses on configuring $idleSlope_{v, \mathrm{BE}}$ for every endpoint $v$, that is, the idle slope that governs BE transmission.

The choice of idle slope directly determines the worst-case number of BE frames that can accumulate at any switch during a BE phase. A larger idle slope allows an endpoint to transmit more frequently, which increases BE throughput but also increases the burst size arriving at downstream switches. A smaller idle slope reduces the burst at the cost of throughput. The BE Shaping Agent must choose an idle slope that keeps the worst-case queue depth bounded.

To formalize this trade-off, we consider the bottleneck link experienced by each endpoint along its BE paths. For an endpoint $v$, let $\mathcal{E}_v$ denote the set of links traversed by any BE flow originating at $v$. Let $N_{\mathrm{BE}}(e) = \sum_{f \in F_{\mathrm{BE}}, e \in R_f} n_f\cdot L_f$ be the total number of BE bits traversing link $e$ per cycle. The bottleneck link of endpoint $v$ is the most congested link along its BE paths. Formally,
\begin{equation}\label{eq:bnv}
b(v) = \arg\min_{e \in \mathcal{E}_v} \frac{C_e}{N_{\mathrm{BE}}(e)}.
\end{equation}

Let $N_{\mathrm{BE}}(v)=\sum_{f \in F_{\mathrm{BE}}, v^{\mathrm{src}}_f=v} n_f\cdot L_f$ be the number of BE bits sent from endpoint $v$ per cycle. The BE Shaping Agent assigns the BE idle slope of endpoint $v$ as
\begin{equation}
idleSlope_{v, \mathrm{BE}} = \frac{N_\mathrm{BE}(v)\cdot C_{b(v)}}{N_{\mathrm{BE}}(b(v))}.
\label{eq:idle_slope}
\end{equation}

This assignment has a clear intuition. At the bottleneck link $b(v)$, the aggregate BE rate of allendpoints sharing this link equals at most $C_{b(v)}$. The link therefore operates never above its capacity. Sustained queue growth at any switch in the network is therefore impossible (as shown below).

The computed idle slopes are pushed to the endpoints together with the GCL produced by the Network Partitioning Agent. Each endpoint enforces the idle slope locally according to the CBS credit model defined in Section~\ref{sec:cbs_model}, ensuring that its BE transmission rate never exceeds its assigned share of its bottleneck link. Computing all idle slopes requires a single pass over the BE flows and their routes and is therefore linear in the total size of the BE routes. The formal justification that the resulting configuration bounds the BE queue depth at every switch, and that the guard band in Equation~\eqref{eq:guard_be} is therefore sufficient to drain the network, is presented next.
\subsubsection{Formal Guarantees}
\label{sec:scenario_a_proof}
We now prove that the combined configuration produced by the three agents in the TNP method guarantees deterministic delivery of all admitted TC flows. The TC Scheduling Agent already eliminates all TC-versus-TC contention by ensuring that no two TC frames occupy the same link simultaneously. What remains to be shown is that BE traffic cannot interfere with TC traffic. Specifically, we must prove that at the end of every guard band, all in-flight BE frames have reached their destinations, so that the network is completely free of BE traffic when the next TC phase begins. The argument proceeds in three steps. The \emph{Setup} shows that no link is overloaded. The \emph{Queue bound} then uses this to bound the number of BE frames queued at any switch port. Proposition~\ref{thm:drain} then uses that bound to show that the guard band $g_{\mathrm{be}}$ of Equation~\eqref{eq:guard_be} fully drains the network before each TC phase begins, and Theorem~\ref{thm:deadline} concludes that all admitted TC flows meet their deadlines.

\paragraph{Setup - links are not overloaded}
Consider any link $e \in E$.
Let $v$ be an endpoint that sends a BE flow through $e$. By construction of the BE Shaping Agent, the idle slope at $v$ is $idleSlope_{v, \mathrm{BE}} = N_{\mathrm{BE}}(v)\cdot C_{b(v)} / N_{\mathrm{BE}}(b(v))$ (Equation \eqref{eq:idle_slope}). Because $b(v)$ is the bottleneck of $v$'s BE paths as defined in Equation \eqref{eq:bnv}, $C_e / N_{\mathrm{BE}}(e) \geq C_{b(v)} / N_{\mathrm{BE}}(b(v))$. Hence,
\begin{equation}
idleSlope_{v, \mathrm{BE}} \leq \frac{N_{\mathrm{BE}}(v)\cdot C_e}{N_{\mathrm{BE}}(e)}.
\label{eq:slope_bound}
\end{equation}

Inequality~\eqref{eq:slope_bound} holds for every endpoint contributing to link $e$. Given that $\Sigma_{v\in V^{ep},f\in F_{\mathrm{BE}},v^{src}_f=v,e\in R_f}N_{\mathrm{BE}}(v)=N_{\mathrm{BE}}(e)$, the aggregate BE rate arriving at $e$ is at most $N_{\mathrm{BE}}(e) \cdot C_e / N_{\mathrm{BE}}(e) = C_e$. The link therefore operates at or below capacity at all times.

\paragraph{Queue bound}
Consider the output port feeding link $e$. At the start of a BE phase the network holds no BE traffic, so its queue builds only from frames released afterward. Two facts bound its depth. First, \emph{bounded simultaneous arrivals}: at most $\beta_{\mathrm{BE}}(e)$ BE frames can be present at the port at once, since no more than $K_{\mathrm{BE}}(e)$ distinct BE flows feed $e$ and no more than $d_{\mathrm{in}}(e)$ frames can arrive simultaneously through the feeding switch's BE-carrying input ports, giving $\beta_{\mathrm{BE}}(e) = \min(K_{\mathrm{BE}}(e), d_{\mathrm{in}}(e))$ as in Equation~\eqref{eq:beta_be}. Second, \emph{bounded refill}: after sending a frame, each source's CBS credit turns negative and must recover before the source may send again, so no source contributes more than one frame to a simultaneous burst. Together with the Setup argument, showing that the aggregate BE arrival rate at $e$ never exceeds the service rate $C_e$, the port clears each burst before a new one can form. From these two facts, at most $\beta_{\mathrm{BE}}(e) - 1$ frames are ever queued ahead of a frame in service.
\begin{proposition}[Network Drain]
\label{thm:drain}
If the guard band duration is set according to Equation~\eqref{eq:guard_be}, then, at the end of the guard band, every BE frame that was in flight at the end of the preceding BE phase has reached its destination.
\end{proposition}

\begin{proof}
Let $T$ denote the end of the BE phase. Any BE frame still in the network at time $T$ must belong to some BE flow $f \in F_{\mathrm{BE}}$. The worst case is a frame that began transmission from its source at time $T$ exactly and must still traverse its entire route $R_f$ for some longest route $R_f$.

Along each link $e \in R_f$, the frame spends its transmission time $\tau_{f,e}$. At each intermediate switch $s \in S_f$, the frame incurs the processing delay $\delta_s$ and may queue behind other BE frames at the output port. By the queue bound, which combines the bounded simultaneous arrivals and bounded refill at the port, at most $\beta_{\mathrm{BE}}(e_{\mathrm{out}}^{s,f}) - 1$ such frames are ahead of it, each of at most $L^{\mathrm{max}}_{\mathrm{BE}}$ bits, so the worst-case queuing delay there is $(\beta_{\mathrm{BE}}(e_{\mathrm{out}}^{s,f}) - 1) \cdot \tau^{\mathrm{BE}}_{e_{\mathrm{out}}^{s,f}}$.

The total travel time of the frame is therefore at most
\begin{equation}
\sum_{e \in R_f} \tau_{f,e}
+ \sum_{s \in S_f} \left( \delta_s + (\beta_{\mathrm{BE}}(e_{\mathrm{out}}^{s,f}) - 1) \cdot \tau^{\mathrm{BE}}_{e_{\mathrm{out}}^{s,f}} \right).
\end{equation}

Taking the maximum of this expression over all BE flows yields exactly $g_{\mathrm{be}}$ of Equation~\eqref{eq:guard_be}. Therefore, by time $T + g_{\mathrm{be}}$, the frame has reached its destination. Since this argument holds for the worst-case BE frame in flight at time $T$, every BE frame in flight at time $T$ has reached its destination by time $T + g_{\mathrm{be}}$.
\end{proof}

\begin{theorem}[TC Deadline Satisfaction]
\label{thm:deadline}
Under the TNP configuration produced by the three agents, every admitted TC flow $f \in F_{\mathrm{TC}}$ meets its end-to-end deadline: $t^{\mathrm{arr}}_{f,i} - t^{\mathrm{gen}}_f \leq d_f$ for every frame $i$ of $f$.
\end{theorem}

\begin{proof}
Proposition~\ref{thm:drain} implies that at the start of every TC phase, no BE frame is in flight anywhere in the network. Combined with the fact that the TC Scheduling Agent produces a collision-free TC schedule, this ensures that every admitted TC frame traverses its route without being delayed by any other traffic. Consequently, every admitted TC flow meets its end-to-end deadline.
\end{proof}

\subsubsection{Running Example}
\label{sec:scenario_a_example}

To make the operation of the three agents concrete, we walk through a small example. Consider the topology shown in Figure~\ref{fig:ilustrative_topology}: four endpoints $v_1, v_2, v_3, v_4$ connected through two switches $s_1$ and $s_2$, with the Central Network Controller attached to $s_2$ as a regular endpoint. The shared link $e_3$ between the two switches is the natural bottleneck of the topology, since every flow crossing from the left side to the right side must traverse it.

Assume all links have capacity $C_e = 1$~Gbps. The switch processing delay is $\delta_s = 1$~µs, the TAS cycle is $T_{\mathrm{cyc}} = 1000$~µs, and every frame has size $L_f = 1500$~B, giving a transmission time of $\tau_{f,e} = 12$~µs on every link. We consider six flows, all transmitting a single frame with transmission period $T_{\mathrm{cyc}}$. Hence, for every flow, $n_f = 1$ and $\tau^{\max}_f = 12$~µs, so each aggregate reservation reduces to a single frame of duration $\Lambda_f = 12$~µs. Four TC flows create path diversity through the bottleneck. Two BE flows share the same fabric. Table~\ref{tab:example_flows} summarizes the flow set.

\begin{table}[h]
\centering
\caption{Flows used in the running example}
\label{tab:example_flows}
\begin{tabular}{ccccc}
\hline
\textbf{Flow} & \textbf{Class} & \textbf{Source $\to$ Destination} & \textbf{Route} & $d_f$ \\
\hline
$f_1$ & TC & $v_1 \to v_3$ & $(e_1, e_3, e_4)$ & 100~µs \\
$f_2$ & TC & $v_2 \to v_4$ & $(e_2, e_3, e_5)$ & 150~µs \\
$f_3$ & TC & $v_1 \to v_4$ & $(e_1, e_3, e_5)$ & 200~µs \\
$f_4$ & TC & $v_2 \to v_3$ & $(e_2, e_3, e_4)$ & 250~µs \\
$f_5$ & BE & $v_1 \to v_3$ & $(e_1, e_3, e_4)$ & --- \\
$f_6$ & BE & $v_2 \to v_4$ & $(e_2, e_3, e_5)$ & --- \\
\hline
\end{tabular}
\end{table}

\paragraph{TC Scheduling Agent}
The four TC flows are processed in EDF order based on the deadlines in Table~\ref{tab:example_flows}, yielding the order $f_1, f_2, f_3, f_4$. Flow $f_1$ is placed first at $t^{\mathrm{rel}}_{f_1} = 0$, occupying link $e_3$ during $[13, 25)$~µs. Flow $f_2$ would collide with $f_1$ on $e_3$ if released at time $0$, so the forbidden-interval scan shifts its release time to $t^{\mathrm{rel}}_{f_2} = 12$~µs, placing its $e_3$ reservation at $[25, 37)$~µs immediately after $f_1$. Flows $f_3$ and $f_4$ have a generation time of $t^{\mathrm{gen}}_{f_3} = t^{\mathrm{gen}}_{f_4} = 500$~µs, forming a second cluster later in the cycle. The TC Scheduling Agent places $f_3$ at $t^{\mathrm{rel}}_{f_3} = 500$~µs and shifts $f_4$ to $t^{\mathrm{rel}}_{f_4} = 512$~µs to avoid a collision on $e_3$. The resulting per-link occupancy is shown in Figure~\ref{fig:example_per_link}. The bottleneck link $e_3$ carries all four TC flows packed back-to-back within each cluster, while other links carry only the subset of flows that traverse them. Since the schedule is contention-free, every TC frame traverses its route in exactly the sum of its transmission and processing delays: the end-to-end latency from release is $\omega_f(3) + \tau_{f,e_3} = 2 \cdot (12 + 1) + 12 = 38$~µs for every TC flow in this topology.

\begin{figure}[h]
\centering
\includegraphics[width=\columnwidth]{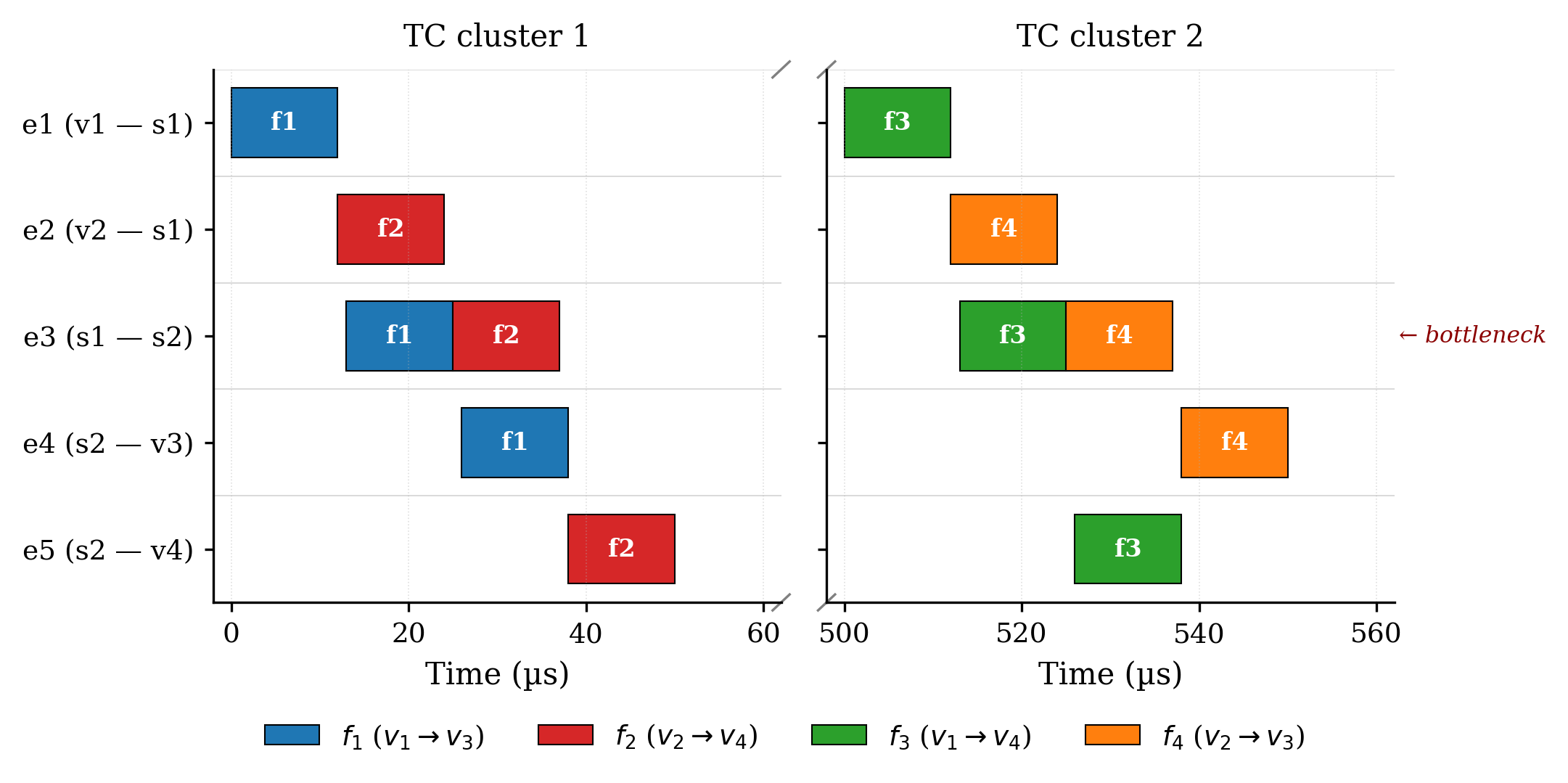}
\caption{Per-link occupancy after the TC Scheduling Agent places the
four TC flows. The broken time axis displays the two TC clusters at a
readable scale.}
\label{fig:example_per_link}
\end{figure}

\paragraph{BE Shaping Agent}
The BE flows $f_5$ and $f_6$ both traverse the bottleneck link $e_3$, so $K_{\mathrm{BE}}(e_3) = 2$, and they carry equal traffic, giving $N_{\mathrm{BE}}(v_1) = N_{\mathrm{BE}}(v_2) = N_{\mathrm{BE}}(e_3)/2$. For all other links, $K_{\mathrm{BE}}(e) = 1$. The bottleneck of both $v_1$ and $v_2$ is therefore $e_3$. The BE Shaping Agent assigns both source endpoints $v_i$, for $i\in\{1,2\}$, an idle slope of
\begin{equation}
idleSlope_{v_i, \mathrm{BE}} = \frac{N_{\mathrm{BE}}(v_i)\,C_{e_3}}{N_{\mathrm{BE}}(e_3)} = \frac{C_{e_3}}{2} = 500~\text{Mbps}.
\end{equation}

The two BE flows reach $e_3$ through the separate input links $e_1$ and $e_2$ of switch $s_1$, so $d_{\mathrm{in}}(e_3) = 2$ and the BE burst bound is $\beta_{\mathrm{BE}}(e_3) = \min(K_{\mathrm{BE}}(e_3), d_{\mathrm{in}}(e_3)) = 2$. The queue depth at any switch output port on the path of a BE flow is therefore bounded by $\beta_{\mathrm{BE}}(e) - 1$ frames.

\paragraph{Network Partitioning Agent}
The Network Partitioning Agent first computes the guard band. The worst-case BE route traverses three links and two switches, with a queue drain contribution only at the switch feeding the bottleneck link $e_3$. Plugging the values into Equation~\eqref{eq:guard_be} gives
\begin{equation}
g_{\mathrm{be}} = 3 \cdot 12 + 2 \cdot 1 + (2-1) \cdot 12 = 50~\text{µs}.
\end{equation}

Next, the agent clusters the TC reservations. The first cluster ends when the last frame of $f_2$ leaves the network: released at $t^{\mathrm{rel}}_{f_2} = 12$~µs, the frame reaches its final link $e_5$ after a cumulative offset of $\omega_{f_2}(3) = 2 \cdot 12 + 2 \cdot 1 = 26$~µs and completes its transmission there at $12 + 26 + 12 = 50$~µs. The gap between the end of the first cluster (at $50$~µs) and the start of the second (at $500$~µs) is $450$~µs, well above the merge threshold $g_{\mathrm{be}} + \Delta_{\mathrm{min}} = 100$~µs, so the two clusters remain separate TC phases. The resulting global GCL is shown in Figure~\ref{fig:example_gcl_cycle}: a TC phase at the start of the cycle carrying $f_1$ and $f_2$, a 400~µs BE phase, a 50~µs guard band, a second TC phase carrying $f_3$ and $f_4$, a second 400~µs BE phase, and a final 50~µs guard band.

\begin{figure}[h]
\centering
\includegraphics[width=\columnwidth]{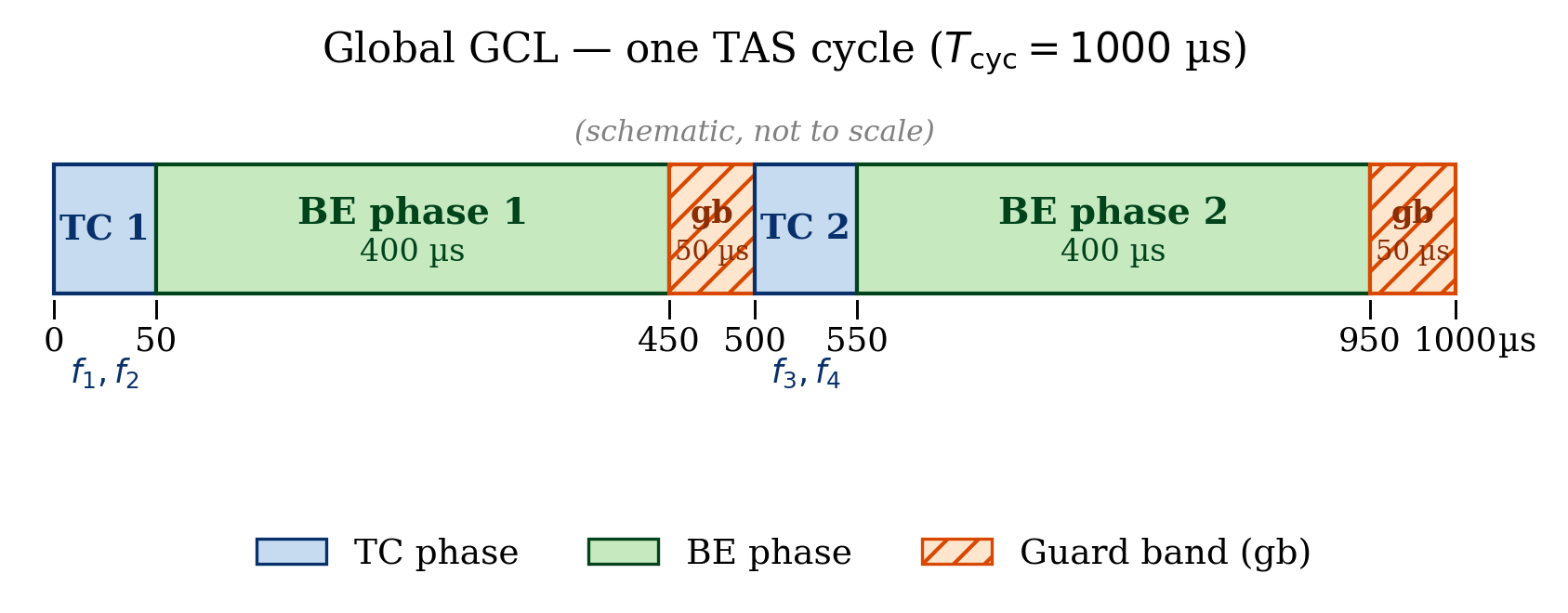}
\caption{The global GCL for the running example.}
\label{fig:example_gcl_cycle}
\end{figure}

The same GCL is pushed to every endpoint. Combined with the release times from the TC Scheduling Agent and the idle slopes from the BE Shaping Agent, this completes the configuration according to the TNP method for the example. Every TC frame is transmitted at its scheduled instant and traverses its route without interference, and each BE endpoint receives a guaranteed share of the bottleneck link during each BE phase.

\subsection{Traffic Prioritization}
\label{sec:scenario_b}

For the TP method, we assume that the network consists of commodity switches with strict-priority queuing at each output port. At any switch, whenever a TC frame is ready for transmission, it is served before any pending BE frame. This property replaces the need for temporal separation between TC and BE traffic, because TC-versus-BE contention is resolved locally at every switch rather than globally by a source GCL.

As a consequence, the central controller employs only two agents: the TC Scheduling Agent and the BE Shaping Agent. The Network Partitioning Agent is no longer needed, and the resulting configuration contains no global GCL, no TC phases, and no guard band. BE traffic is therefore \textit{work-conserving}: a BE endpoint may transmit whenever its CBS credit is non-negative and its egress link is idle, without waiting for a scheduled BE phase. This improves overall link utilization compared to the TNP method at the cost of requiring switches to support strict-priority queuing.

The absence of a global GCL does not eliminate the need for release-time coordination. Strict priority resolves contention between TC and BE traffic, but it does not resolve contention between two TC frames arriving simultaneously at the same switch output port. The TC Scheduling Agent therefore continues to produce a collision-free TC schedule, using the same forbidden-interval scan introduced in Section~\ref{sec:release_time_agent}. The only difference is that the per-hop transmission time used by the scan is inflated to account for the worst-case blocking incurred by a TC frame arriving at a switch where a BE frame has just begun transmission on the output link. Because frames are non-preemptive, a newly arriving TC frame may wait for up to one maximum-size BE frame to complete before the switch serves the TC frame. This blocking term is bounded, finite, and independent of the number of BE flows, which makes the worst-case TC latency a closed-form quantity rather than one that depends on queue dynamics.

\subsubsection{TC Scheduling Agent}
\label{sec:release_time_agent_b}

The TC Scheduling Agent in the TP method follows the same forbidden-interval scan structure as in TNP method (Algorithm~\ref{alg:rt_agent}), but with two modifications that reflect the strict-priority behavior of the switches.

First, the per-link time consumed by a TC frame along its route must account for the worst-case BE blocking incurred at every switch. When a TC frame arrives at a switch $s$ whose output link $e_{\mathrm{out}}^{s,f}$ is currently transmitting a BE frame, the TC frame must wait for that BE frame to complete. Because the transmission is non-preemptive, this wait is bounded by the transmission time of a maximum-size BE frame on that link:
\begin{equation}
\tau^{\mathrm{BE}}_e = \frac{L^{\mathrm{max}}_{\mathrm{BE}}}{C_e},
\end{equation}

where $L^{\mathrm{max}}_{\mathrm{BE}}$ is the maximum BE frame size. The cumulative offset from the source to link $e_k$ along the route $R_f$ is therefore updated to include this blocking term at every intermediate switch:
\begin{equation}
\omega^{B}_f(k) = \sum_{i=1}^{k-1} \left( \tau_{f,e_i} + \delta_{s_i} + \tau^{\mathrm{BE}}_{e_{i+1}} \right).
\end{equation}

This modified offset replaces $\omega_f(k)$ of Equation \eqref{eq:offset} in the forbidden-interval scan (line~6 of Algorithm \ref{alg:rt_agent}). As in the TNP method, the $n_f$ frames of one period are released with spacing $\tau^{\max}_f$ and represented as one aggregate reservation of duration $\Lambda_f$, and each flow is placed with all its $T_{\mathrm{cyc}}/T_f$ instances per cycle: the aggregate reservation of the $q$-th instance on link $e_k$ becomes $[t^{\mathrm{rel}}_f + q \cdot T_f + \omega^{B}_f(k),\; t^{\mathrm{rel}}_f + q \cdot T_f + \omega^{B}_f(k) + \Lambda_f)$, and the forbidden ranges, shifted by $-q \cdot T_f$ per instance, are computed accordingly.

Second, the admission condition is identical to that of the TNP method: Equation~\eqref{eq:admission_tnp} applies unchanged, with $\omega^{B}_f$ taking the place of $\omega_f$, and is checked when the release time is selected (line~7 of Algorithm~\ref{alg:rt_agent}). The condition remains valid for bursts under strict priority: every frame of the burst may be blocked by at most one maximum-size BE frame at each switch on its route, a delay captured by $\omega^{B}_f$, and even when blocking delays an earlier frame and a later frame temporarily queues behind it, the spacing $\tau^{\max}_f$ guarantees that each frame $i$ completes its transmission on every link $e_k$ no later than $t^{\mathrm{rel}}_{f,i} + \omega^{B}_f(k) + \tau_{f,e_k}$. 

The admission condition of TP is stricter than that of TNP. In TNP, the end-to-end latency of an admitted flow consists only of its transmission and switch processing delays, since the schedule guarantees a contention-free traversal. In TP, the worst-case latency additionally includes the BE blocking term $\tau^{\mathrm{BE}}_{e_{\mathrm{out}}^{s,f}}$ at every switch on the route. TNP can therefore admit TC flows with tighter deadlines than TP, at the cost of requiring the temporal partitioning of the network. 

All other aspects of the algorithm remain identical to Algorithm~\ref{alg:rt_agent}: the flows are processed in EDF order, the per-link reservation timelines $T_e$ are updated after each placement, and the output is a collision-free set of release times for all admitted TC flows.

After the agent completes, the per-link timelines jointly define a TC schedule in which no two TC frames occupy the same link at the same time, even under the worst-case BE blocking at every switch. Because strict priority at switches ensures that TC is never delayed by any factor beyond this bounded blocking, every admitted TC frame is delivered within its deadline.

\subsubsection{BE Shaping Agent}
\label{sec:cbs_agent_b}

The BE Shaping Agent in the TP method computes the BE idle slope $idleSlope_{v, \mathrm{BE}}$ for every endpoint $v$. Unlike in the TNP method, this configuration no longer needs to bound the BE queue depth at any switch, because strict-priority queuing ensures that TC traffic is never displaced by BE traffic at any output port. The idle slope in the TP method therefore serves a simpler purpose: it provides a fair share of the available BE capacity to each endpoint and prevents any single source from monopolizing the network.

Because no queue-bound constraint must be satisfied, the BE Shaping Agent has greater freedom in how it distributes BE capacity among endpoints. Any assignment that respects link capacities is admissible. For consistency with the TNP method and to preserve the same notion of per-endpoint fairness, we adopt the bottleneck-aware formula already introduced in Equation~\eqref{eq:idle_slope}.

In practice, any higher idle slope values may also be used without compromising TC determinism, since strict priority at switches prevents BE traffic from interfering with TC delivery regardless of the idle slope choice. This flexibility is a practical advantage of the TP method compared to TNP, particularly when BE throughput is more important than strict fairness among endpoints. The only constraints on the idle slopes are that they remain non-negative and do not exceed the link capacity.

\subsubsection{Running Example}
\label{sec:scenario_b_example}

To illustrate how the TP method operates, we revisit the example from Section~\ref{sec:scenario_a_example}. The topology, link capacities, switch processing delays, cycle duration, and the six flows remain identical to those in Figure~\ref{fig:ilustrative_topology} and Table~\ref{tab:example_flows}. The only change is that the switches now have strict-priority queuing.

\paragraph{TC Scheduling Agent}
The forbidden-interval scan proceeds in EDF order as in the TNP method, but the per-hop offset used by the scan is now $\omega^{B}_f(k)$, which includes the worst-case BE blocking term $\tau^{\mathrm{BE}}_e$ at every switch. With a maximum BE frame size of $L^{\mathrm{max}}_{\mathrm{BE}} = 1500$~B and link capacity of $1$~Gbps, we have $\tau^{\mathrm{BE}}_e = 12$~µs on every link.

For flow $f_1$ released at $t^{\mathrm{rel}}_{f_1} = 0$, the reservation on $e_1$ covers $[0, 12)$~µs, but the reservation on $e_3$ is shifted to $[25, 37)$~µs because of the potential blocking at switch $s_1$. Similarly, the reservation on $e_4$ falls at $[50, 62)$~µs. Flow $f_2$ is then placed with $t^{\mathrm{rel}}_{f_2} = 12$~µs, packing its $e_3$ reservation back-to-back with $f_1$ at $[37, 49)$~µs. The second cluster is built analogously: $f_3$ is released at $500$~µs and $f_4$ at $512$~µs. The resulting per-link occupancy is shown in Figure~\ref{fig:example_per_link_b}.

\begin{figure}[h]
\centering
\includegraphics[width=\columnwidth]{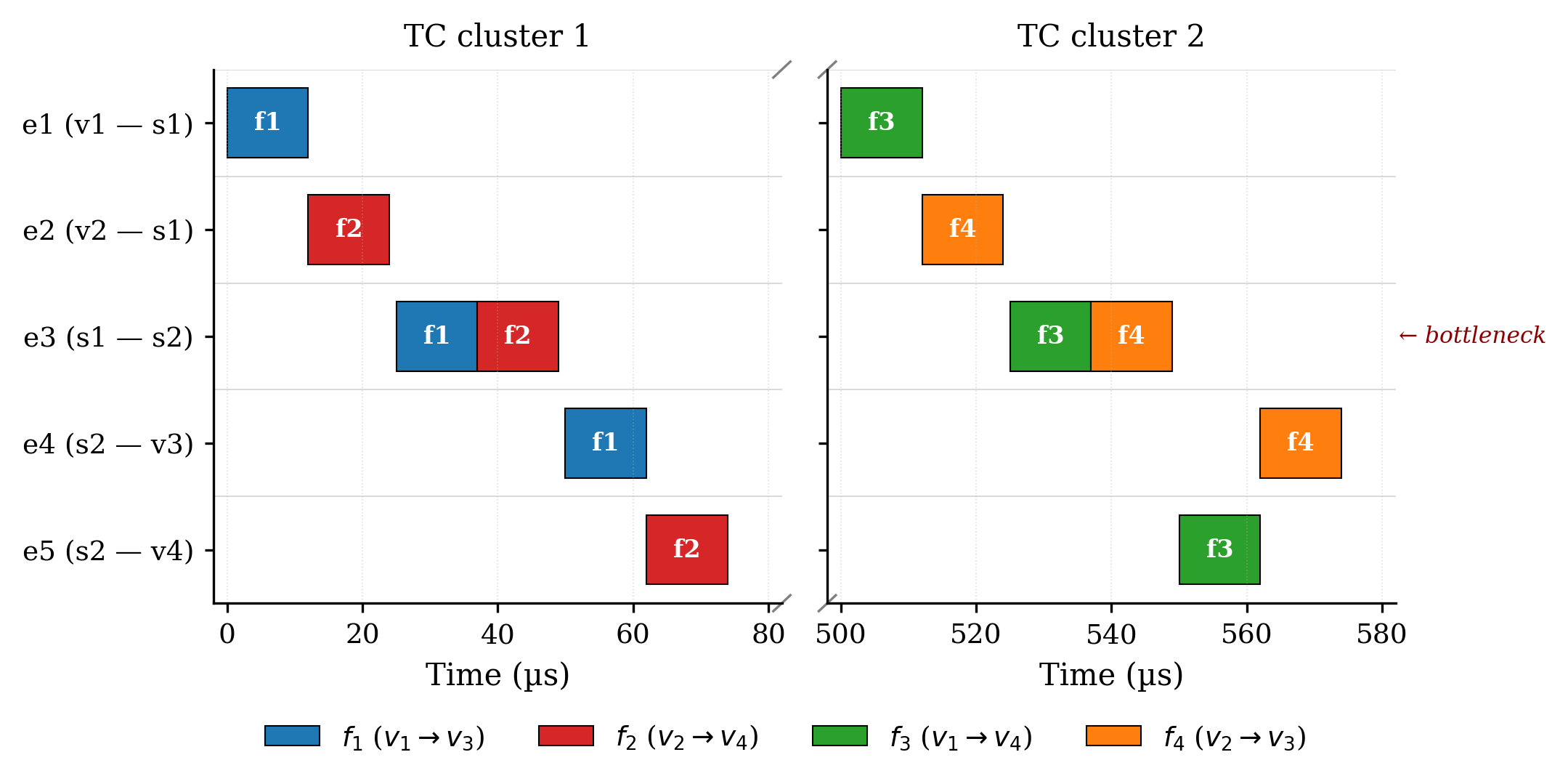}
\caption{Per-link occupancy under TP method.}
\label{fig:example_per_link_b}
\end{figure}

Comparing this result to the result of the TNP method, the TC reservations on the bottleneck link $e_3$ remain back-to-back within each cluster, but each cluster is now 74~µs wide rather than 50~µs. This 48\% increase reflects the worst-case BE blocking term accounted for by the TC Scheduling Agent. The worst-case end-to-end latency for any TC flow in this topology is $D^{\max}_f = 3 \cdot 12 + 2 \cdot (1 + 12) = 62$~µs, which is well within the deadlines of the admitted flows. Compared to the $38$~µs latency of the TNP method, the $24$~µs increase corresponds exactly to the worst-case BE blocking of $12$~µs at each of the two switches, illustrating the latency cost of the work-conserving BE traffic.

\paragraph{BE Shaping Agent}
The set of BE flows and their routes are unchanged, so the bottleneck analysis remains the same as in the TNP method. The BE Shaping Agent therefore assigns $idleSlope_{v_1, \mathrm{BE}} = idleSlope_{v_2, \mathrm{BE}} = 500$~Mbps, matching the TNP method.

\section{Evaluation and Results}\label{sec:eval}
In this section, we evaluate the TNP and TP methods proposed in Section~\ref{sec:SbDN}. Since we introduce a source-based architecture that replaces per-switch scheduling with centralized agents and endpoint enforcement, the primary objective of our evaluation is to demonstrate that both methods effectively protect TC streams and guarantee their deadline requirements. We compare our methods against two baselines: HERMES~\cite{Bujosa2022}, a mixed-criticality TSN scheduler that configures TAS and CBS at every switch port, and LCDN~\cite{Diederich2025}, a recent approach that provides deterministic guarantees on commodity strict-priority switches using network calculus and token bucket shaping. Together, these baselines bracket the design space that SbDN occupies: HERMES represents the standard TSN approach and quantifies the cost of moving enforcement from switches to source endpoints, while LCDN represents the closest alternative for determinism on commodity hardware. We first analyze the computational time required by each agent to produce a valid scheduling configuration. We then compare all four approaches in terms of TC deadline satisfaction and BE throughput under increasing traffic loads.

\subsection{Experimental Setup}
\label{sec:eval_setup}

Our evaluation uses the INSIM simulation platform~\cite{Karimi2025INSIM}. Since our system model and problem definition differ from the standard per-switch TSN formulation, we extended INSIM with additional modules to support source-based TAS and CBS enforcement at endpoints, the three-agent scheduling pipeline of SbDN, and the  HERMES and LCDN baselines, providing a unified evaluation framework. We evaluate all methods on two representative topologies shown in Figure~\ref{fig:topologies}: a mesh topology and a tree topology. The mesh topology creates multiple overlapping paths and high contention on shared links, while the tree topology produces a natural bottleneck at the root. All links operate at 100~Mbps, the switch processing delay is 10~µs. Routes for all flows are given as input; routing is not part of the problem addressed in this work.

\begin{figure}[t]
\centering
\begin{subfigure}[b]{\columnwidth}
\centering
\includegraphics[width=0.75\columnwidth]{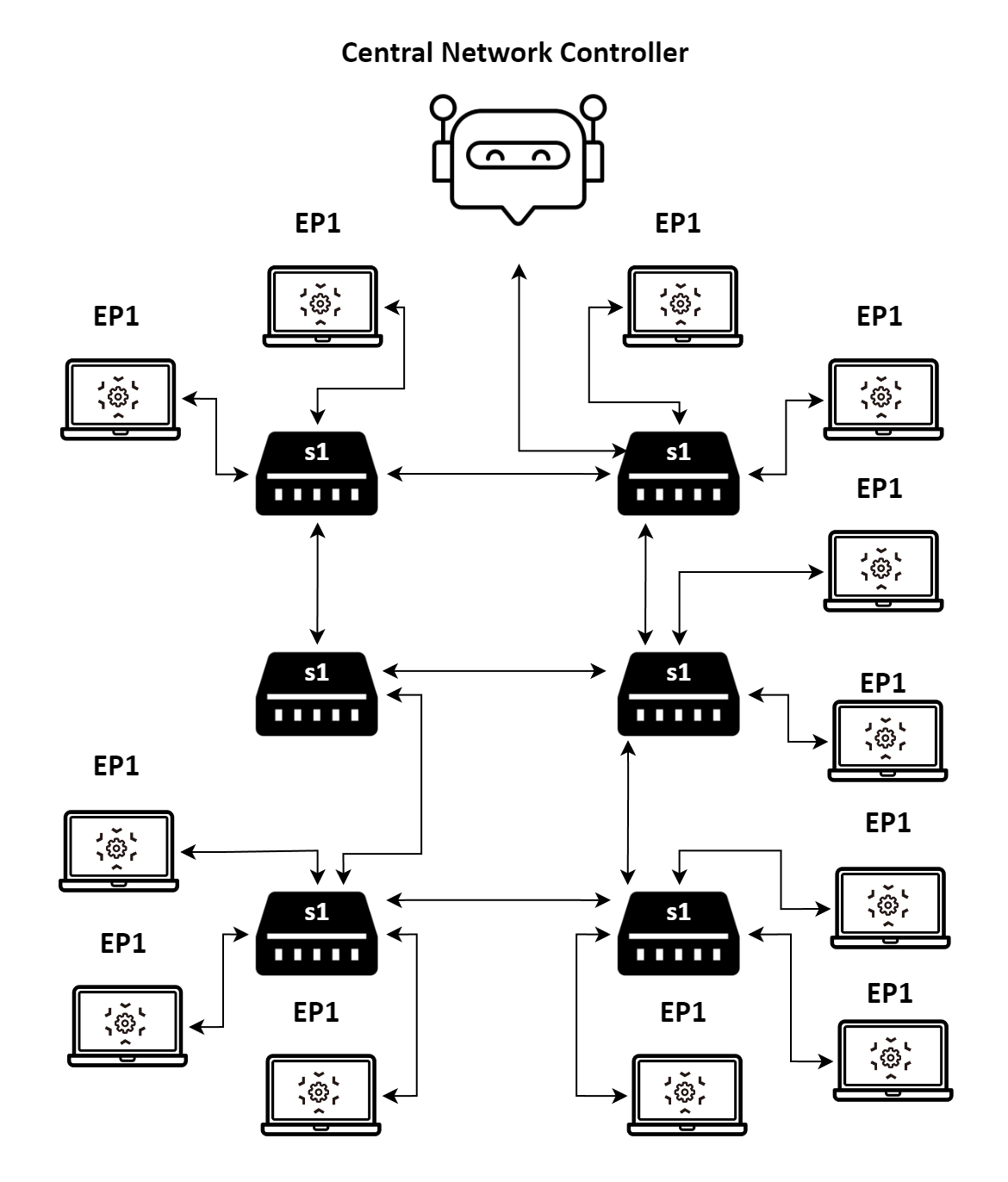}
\caption{Mesh topology.}
\label{fig:top_mesh}
\end{subfigure}

\vspace{0.3cm}

\begin{subfigure}[b]{\columnwidth}
\centering
\includegraphics[width=0.75\columnwidth]{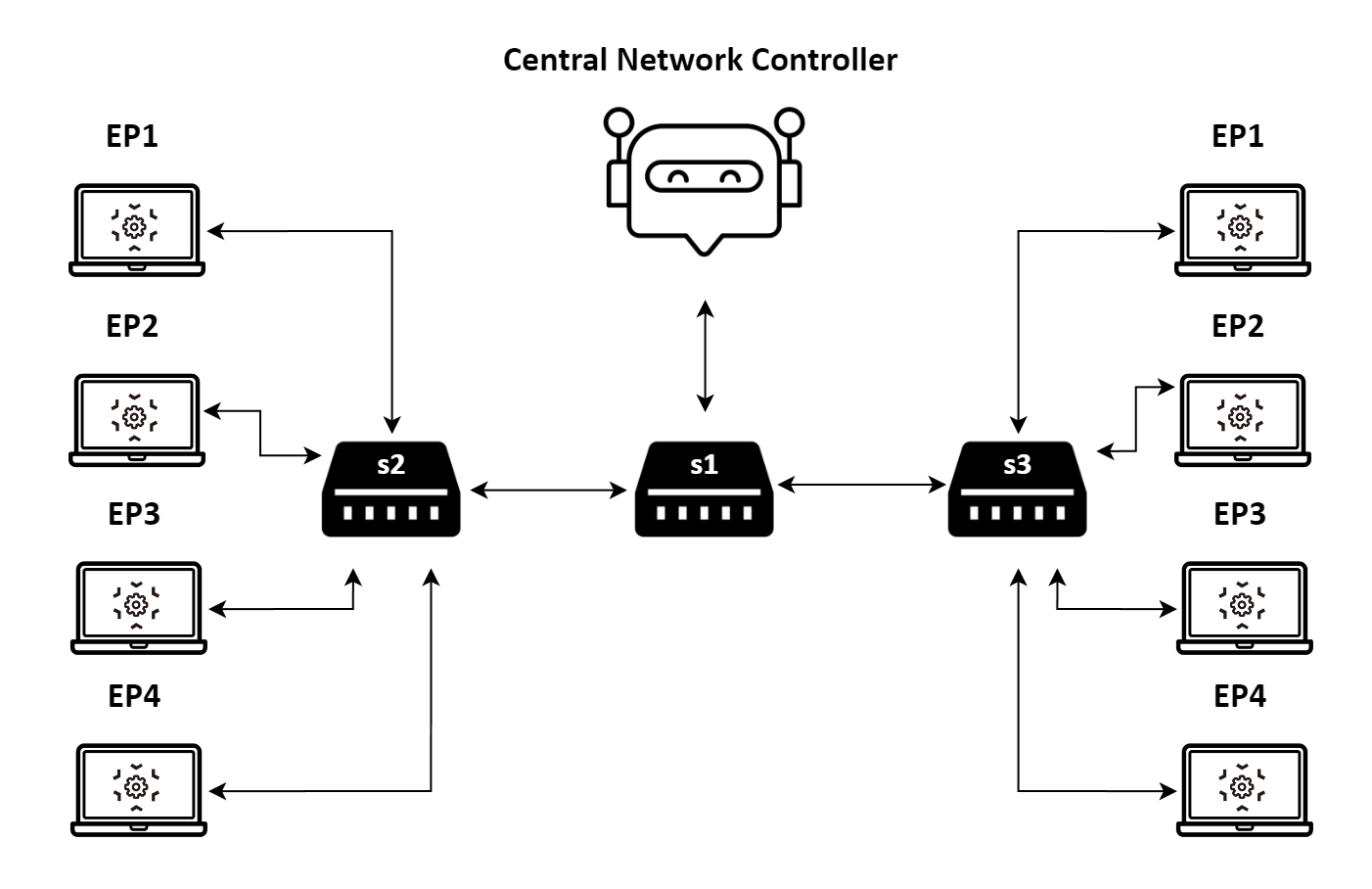}
\caption{Tree topology.}
\label{fig:top_tree}
\end{subfigure}
\caption{Topologies used in the evaluation.}
\label{fig:topologies}
\end{figure}

To systematically cover a range of operating conditions, we design a benchmark that sweeps 20 workload scenarios along two axes: network load heaviness and traffic class balance. The load heaviness axis (H1--H5) varies the total number of flows from 10 to 100, while the balance axis (B1--B4) varies the fraction of TC flows from 20\% to 80\%. Table~\ref{tab:scenarios_grid} lists the resulting TC and BE flow counts for each combination. Each of the 20 workload scenarios is evaluated on both topologies, yielding 40 network scenarios (workload, topology) per method. The complete benchmark is released alongside the SbDN implementation in the same open-source repository.

\begin{table}[t]
\centering
\caption{Scenario grid (TC / BE flow counts)}
\label{tab:scenarios_grid}
\begin{tabular}{lcccc}
\hline
 & \textbf{B1 (20\%)} & \textbf{B2 (40\%)} & \textbf{B3 (60\%)} & \textbf{B4 (80\%)} \\
\hline
\textbf{H1} (10)  & 2\,/\,8  & 4\,/\,6   & 6\,/\,4   & 8\,/\,2 \\
\textbf{H2} (20)  & 4\,/\,16 & 8\,/\,12  & 12\,/\,8  & 16\,/\,4 \\
\textbf{H3} (40)  & 8\,/\,32 & 16\,/\,24 & 24\,/\,16 & 32\,/\,8 \\
\textbf{H4} (70)  & 14\,/\,56 & 28\,/\,42 & 42\,/\,28 & 56\,/\,14 \\
\textbf{H5} (100) & 20\,/\,80 & 40\,/\,60 & 60\,/\,40 & 80\,/\,20 \\
\hline
\end{tabular}
\end{table}

\subsection{Scheduling Time Analysis}
\label{sec:eval_timing}

One motivation for SbDN is enabling runtime reconfiguration when network conditions change. For this to be practical, the scheduling algorithms must produce valid configurations fast enough to support online adaptation. We therefore begin by analyzing the computation time required by each method.

\subsubsection{Per-Agent Breakdown}

Figures~\ref{fig:tnp_breakdown} and~\ref{fig:tp_breakdown} show the total scheduling time of the TNP and TP methods, respectively, broken down by agent for each of the 20 workload scenarios on both topologies. In both methods, the TC Scheduling Agent dominates the computation time, as expected: it performs the forbidden-interval scan over all TC flows and must query and update the per-link interval trees for every candidate placement. The Network Partitioning Agent and the BE Shaping Agent contribute negligibly to the total, consistent with the complexity analysis in Section~\ref{sec:SbDN}: GCL synthesis is log-linear in the total number of TC reservations per cycle, idle slope computation is linear in the total size of the BE routes, and the forbidden-interval scan of the TC Scheduling Agent dominates.

\begin{figure}[t]
\centering
\begin{subfigure}[b]{\columnwidth}
\centering
\includegraphics[width=\columnwidth]{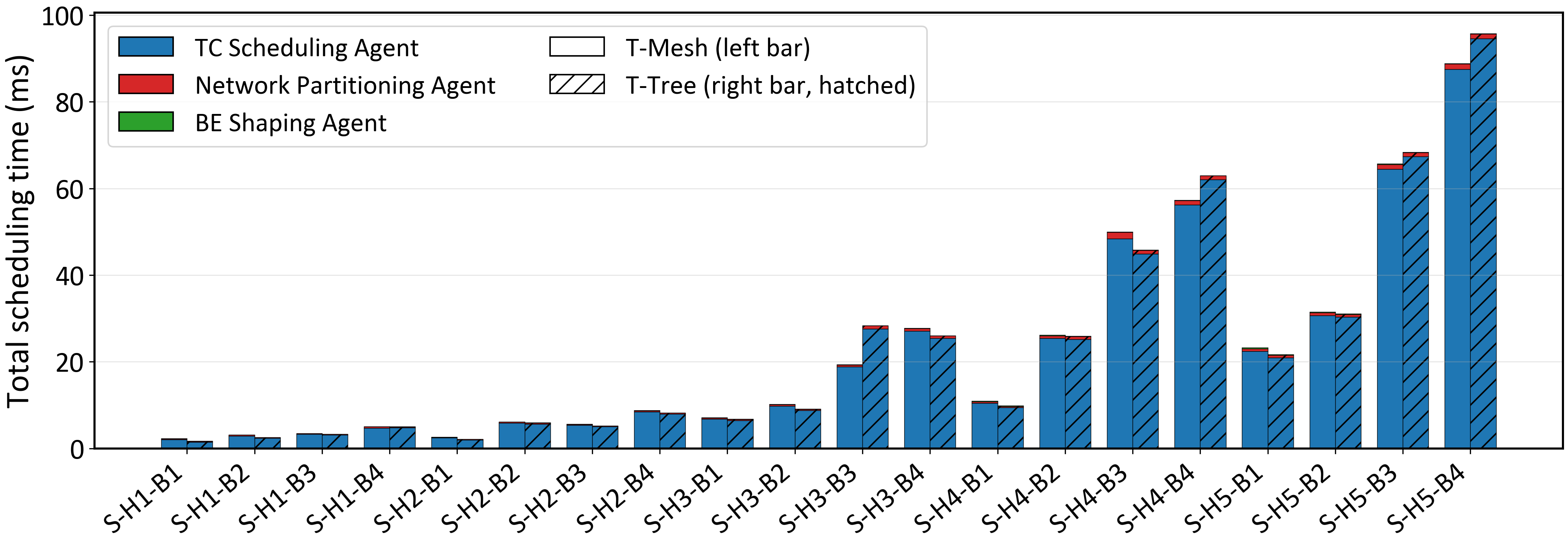}
\caption{TNP.}
\label{fig:tnp_breakdown}
\end{subfigure}

\vspace{0.2cm}

\begin{subfigure}[b]{\columnwidth}
\centering
\includegraphics[width=\columnwidth]{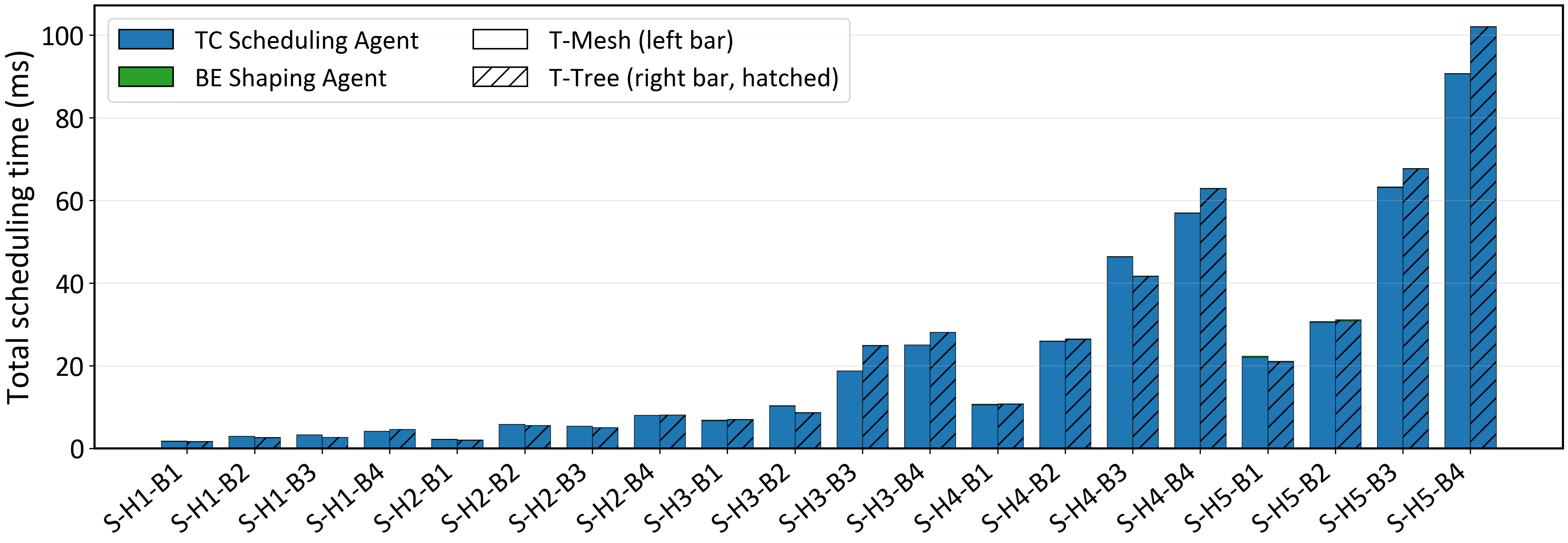}
\caption{TP.}
\label{fig:tp_breakdown}
\end{subfigure}
\caption{Per-agent scheduling time breakdown across all 20 scenarios on
both topologies.}
\label{fig:agent_breakdown}
\end{figure}

Because the TC Scheduling Agent accounts for almost all of the runtime and is essentially the same in both methods, TNP and TP exhibit nearly identical scheduling times across all scenarios; the Network Partitioning Agent used only by TNP adds negligible overhead. Both methods remain in the low millisecond range, from under 2~ms for the lightest scenarios (H1) to about 100~ms for the heaviest (H5-B4 with 80 TC flows on the tree topology). Scheduling time grows with both the total number of flows and the TC fraction, as a higher TC count increases the number of forbidden intervals that must be evaluated during each placement. Reported values are averaged over 10 runs to remove run-to-run measurement noise.

\subsubsection{Comparison Across Methods}

Figure~\ref{fig:total_box} compares the total scheduling time of TNP, TP, LCDN, and HERMES across all 40 network scenarios. The $y$-axis uses a logarithmic scale because the methods span roughly four orders of magnitude, from sub-millisecond computation to multi-second runtimes. Each dot represents one of the 20 workload scenarios for a given topology, the thick horizontal bar marks the median, and the thin vertical line spans the range from minimum to maximum. We use individual data points rather than conventional box plots because the interquartile range rectangle conveys little useful information on a logarithmic scale, where visual area does not correspond to absolute spread.

\begin{figure}[t]
\centering
\includegraphics[width=\columnwidth]{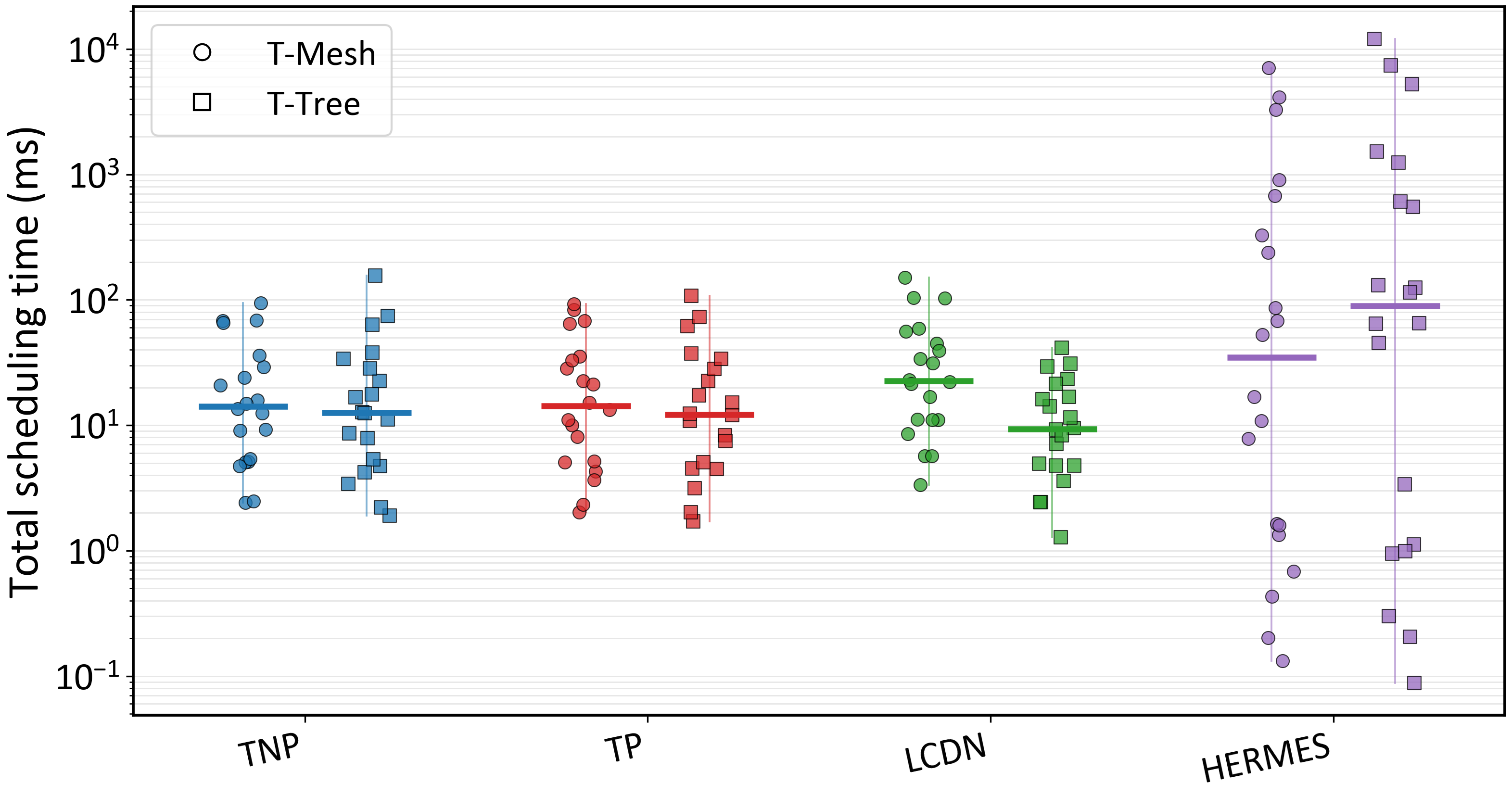}
\caption{Total scheduling time per method across all 20 scenarios on
both topologies (log scale).}
\label{fig:total_box}
\end{figure}

\begin{figure*}[t]
\centering
\begin{subfigure}[t]{0.48\textwidth}
\centering
\includegraphics[width=\linewidth]{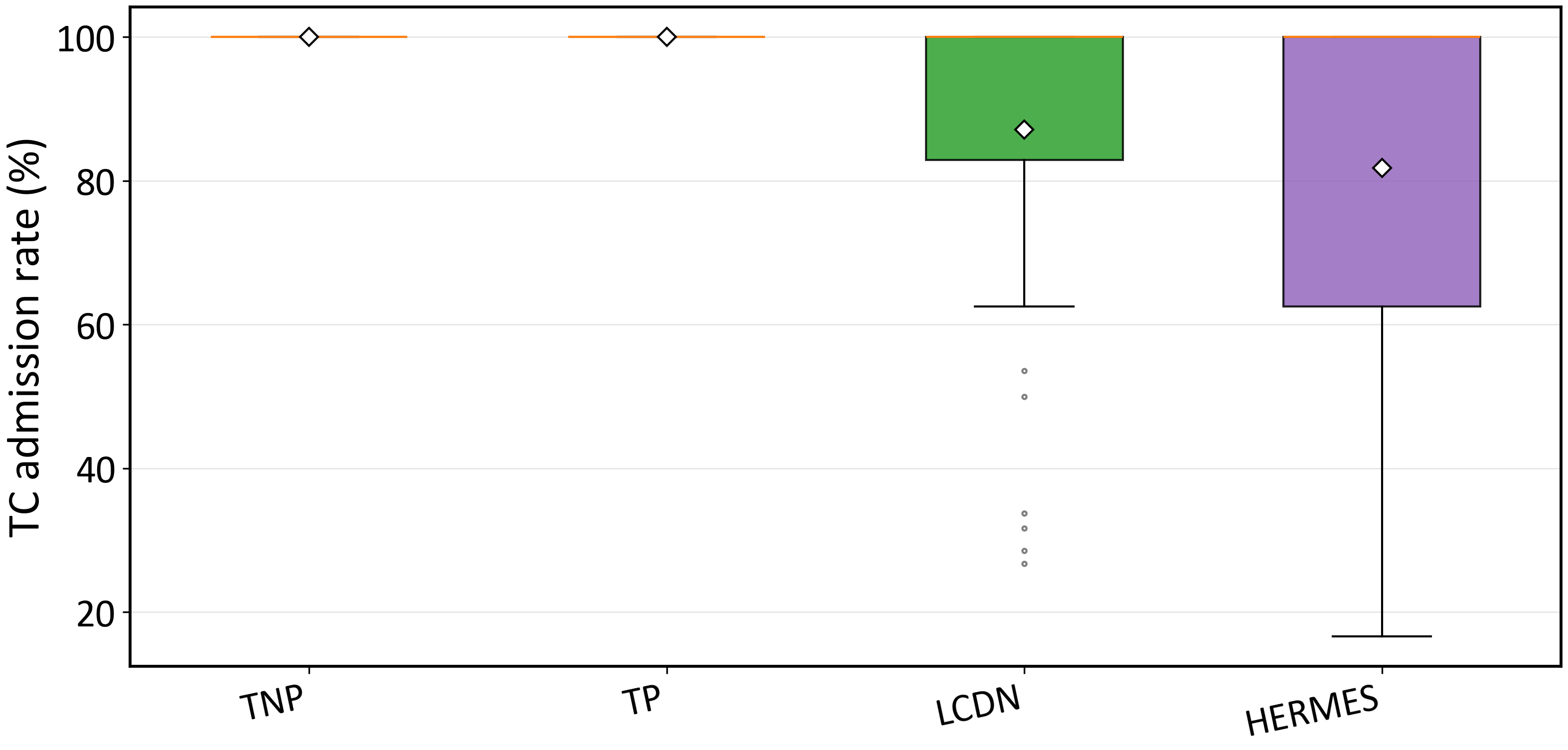}
\caption{TC admission rate across all 40 scenarios.}
\label{fig:admission}
\end{subfigure}
\hfill
\begin{subfigure}[t]{0.48\textwidth}
\centering
\includegraphics[width=\linewidth]{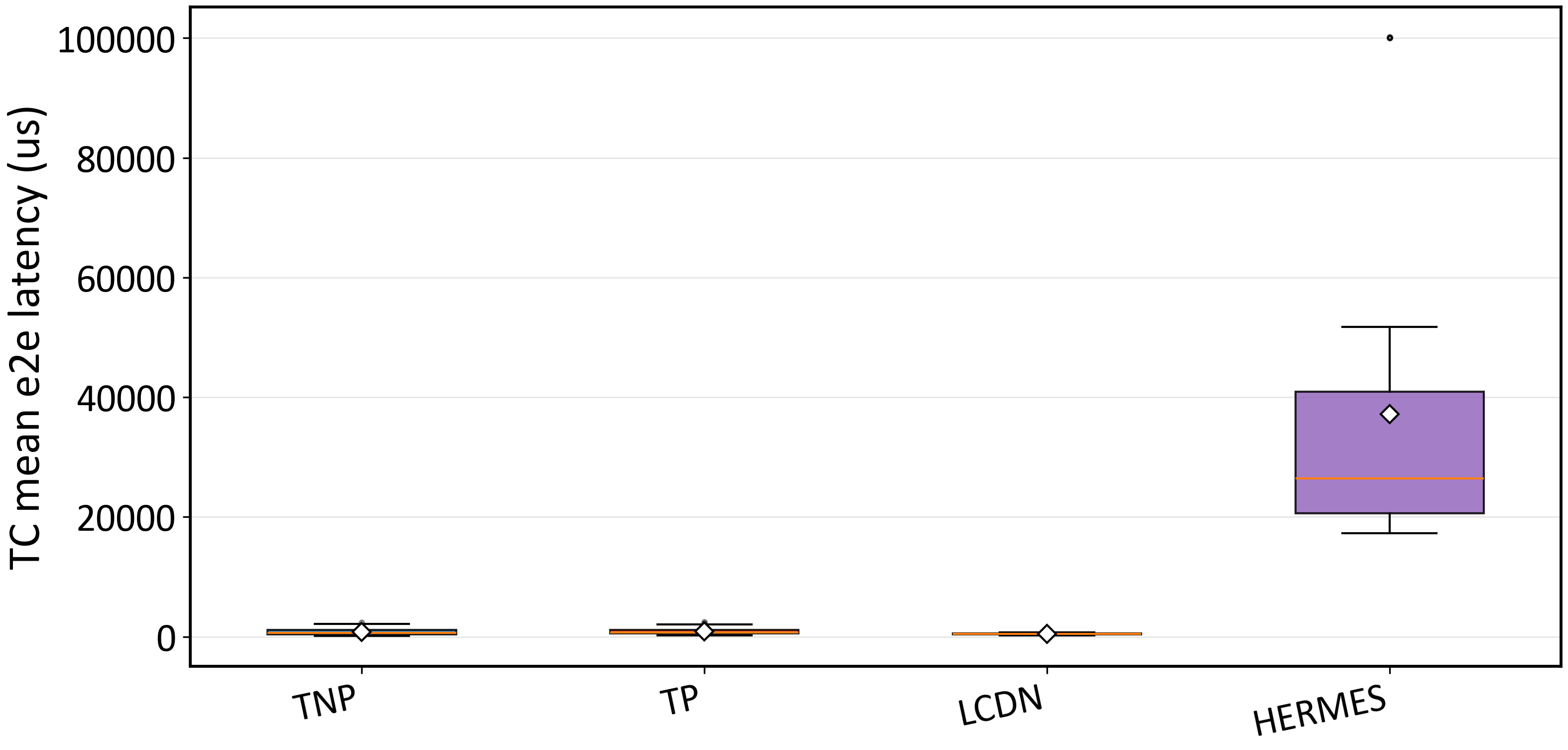}
\caption{Mean end-to-end latency of admitted TC flows across all 40 scenarios.}
\label{fig:tc_e2e}
\end{subfigure}

\vspace{0.2em}

\begin{subfigure}[t]{0.48\textwidth}
\centering
\includegraphics[width=\linewidth]{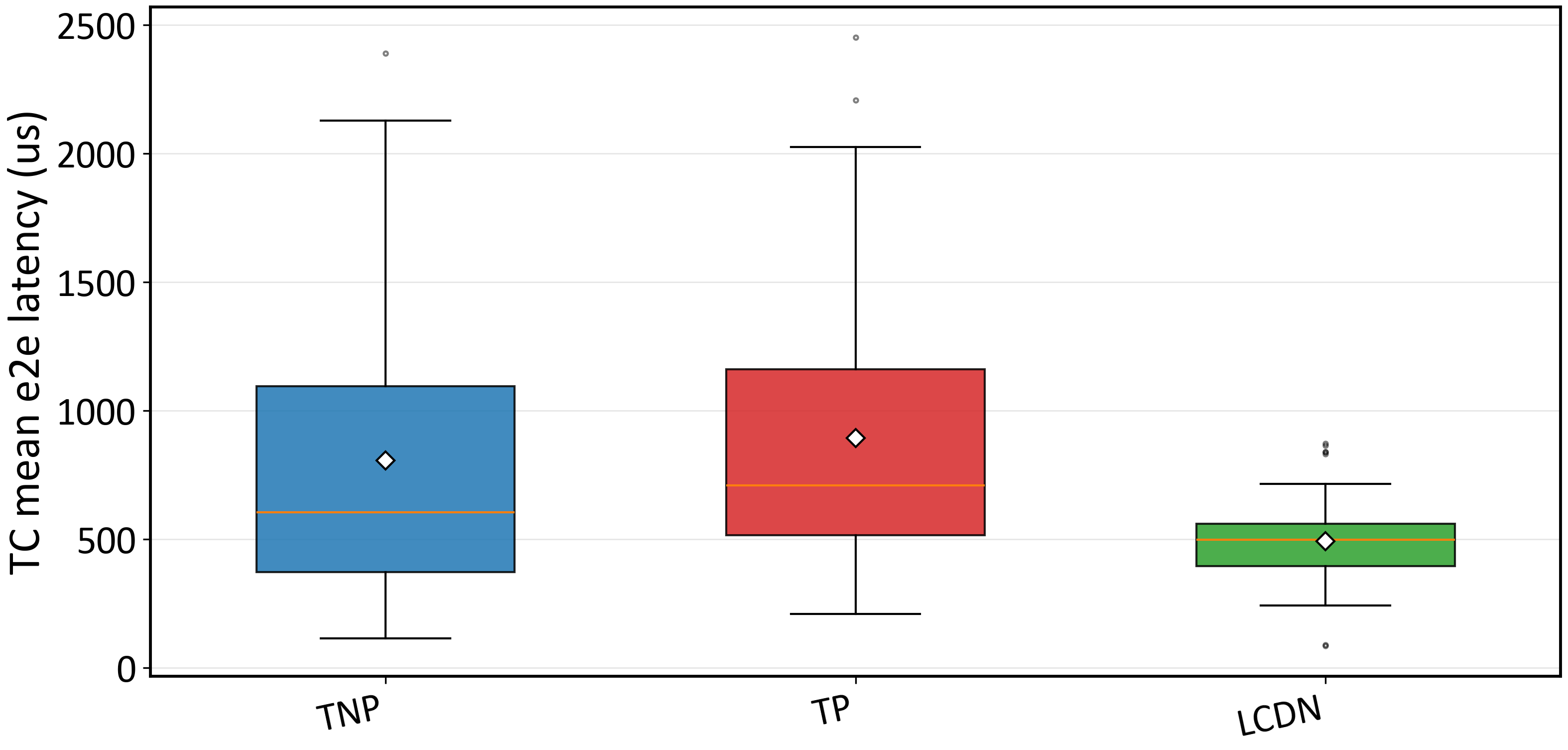}
\caption{Mean end-to-end latency of admitted TC flows (HERMES excluded for readability).}
\label{fig:tc_e2e_zoom}
\end{subfigure}
\hfill
\begin{subfigure}[t]{0.48\textwidth}
\centering
\includegraphics[width=\linewidth]{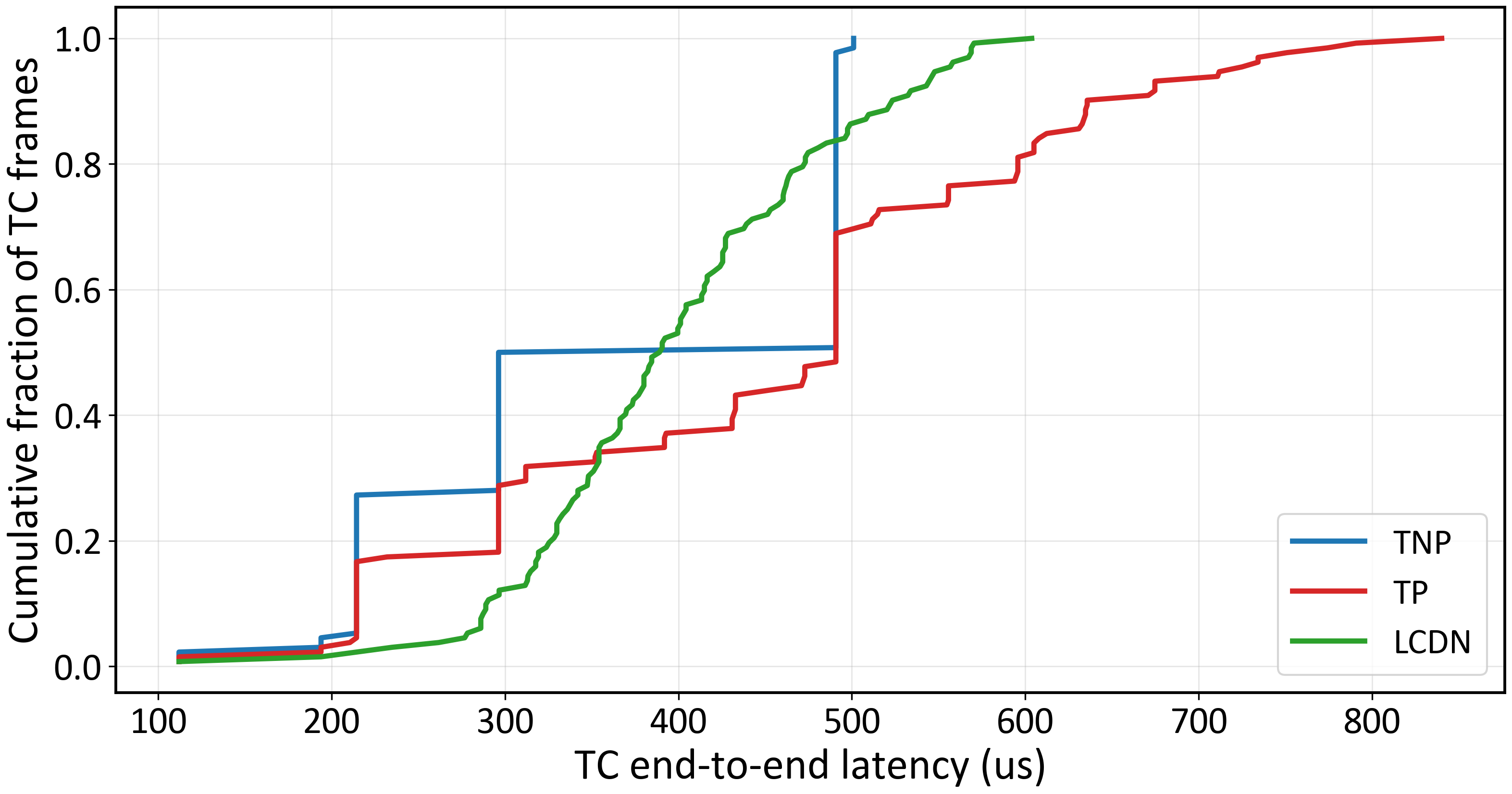}
\caption{CDF of per-frame TC end-to-end latency for scenario H2-B2 on the tree topology.}
\label{fig:tc_cdf}
\end{subfigure}

\caption{TC schedule quality across all 40 scenarios.}
\label{fig:tc_quality}
\end{figure*}

TNP and TP exhibit comparable median scheduling times of approximately 13~ms and remain consistently below 110~ms even in the most demanding scenarios. LCDN achieves similar median times, which is expected since its admission control also performs per-flow analysis with lightweight network calculus computations. HERMES, by contrast, produces median scheduling times that are roughly one order of magnitude higher, with worst-case runtimes exceeding 10~seconds. This difference stems from the algorithmic nature of HERMES, which iteratively slides and adjusts TC frame placements across multiple switch ports to find feasible GCL configurations, a process whose cost compounds with the number of switches and flows.

These results confirm that both TNP and TP are suitable for runtime reconfiguration: a new schedule can be computed in tens of milliseconds, which is well within the time budget available for online adaptation in most industrial and automotive applications.

\subsection{TC Schedule Quality}
\label{sec:eval_tc}

We evaluate the quality of the TC schedule produced by each method along two interlinked dimensions: the fraction of TC flows that can be admitted into the network, and the end-to-end latency experienced by the admitted flows. These two metrics must be read together: a method that rejects difficult flows trivially achieves lower latency on the ones it admits, so latency results are only meaningful in the context of the corresponding admission rate.

All box plots in this section follow the same convention. Each box summarizes a method's distribution of values across the 40 (workload~$\times$~topology) scenarios. The box spans the interquartile range (IQR): its bottom and top edges correspond to the 25th and 75th percentiles, so the middle 50\% of measurements lie inside it. The horizontal line inside the box marks the median, and the white diamond marks the mean. When the diamond sits noticeably above the median line, a small number of high-value scenarios are pulling the average upward. The whiskers extend to the most extreme values within 1.5~$\times$~IQR of the box edges; dots beyond the whiskers are outliers representing configurations whose value is unusually far from the bulk. A short box indicates consistent behavior across the benchmark, while a tall box indicates strong dependence on the workload.

Figure~\ref{fig:tc_quality} summarizes both metrics. Panel~(a) shows the TC admission rate, panel~(b) shows the mean end-to-end latency of admitted TC flows for all four methods, panel~(c) re-plots the latency with HERMES excluded so that the other three methods are visible at a useful scale, and panel~(d) shows the per-frame latency CDF for a representative scenario in which all three commodity-switch methods admit every TC flow.

\textbf{Admission rate.} Both TNP and TP achieve 100\% admission across all 40 scenarios (Fig.~\ref{fig:tc_quality}(a)): every TC flow is successfully placed without violating any deadline constraint. LCDN and HERMES, by contrast, exhibit lower and more variable admission rates. LCDN achieves a median of 100\% but drops below 60\% in several heavy-load scenarios, with outliers as low as 27\%. HERMES shows a wider spread, with a median of 100\% but a mean of approximately 82\% and worst-case admission below 20\%. The lower admission rates of LCDN and HERMES stem from fundamentally different causes: LCDN's network-calculus-based admission control rejects flows whose worst-case delay bounds exceed the deadline, even when the actual delay would be acceptable, while HERMES's iterative slide-and-bump placement can fail to find a feasible GCL configuration under high contention.

\textbf{End-to-end latency.} Figure~\ref{fig:tc_quality}(b) compares the mean end-to-end latency of admitted TC flows across all four methods. TNP, TP, and LCDN all achieve mean latencies in the low hundreds of microseconds, while HERMES exhibits substantially higher latencies with a median around 27~ms and individual scenarios exceeding 100~ms. This result reflects a deliberate design choice in HERMES: its scheduling algorithm prioritizes maximizing BE throughput by placing TC frames in positions that leave the largest possible gaps for BE traffic, which comes at the expense of TC latency. In contrast, the SbDN methods prioritize earliest placement of TC frames through the EDF-ordered forbidden-interval scan.

Because the scale of the HERMES latencies compresses the other three methods into a narrow band near zero, Fig.~\ref{fig:tc_quality}(c) shows the same comparison with HERMES excluded. TNP achieves a median mean latency of approximately 600~µs, while TP is slightly higher at approximately 720~µs. This difference is expected: the TP method inflates the per-hop offset at every switch by the worst-case BE blocking term $\tau^{\mathrm{BE}}_e$, which shifts each TC frame later on the timeline even though the actual blocking may not occur.

LCDN appears at first glance to achieve the lowest mean latency, with a median around 480~µs and a tighter distribution. This result, however, must be read together with Fig.~\ref{fig:tc_quality}(a): LCDN admits fewer TC flows in many scenarios, and the flows it rejects tend to be precisely those with longer routes and tighter deadlines that would have contributed higher latencies to the distribution. The lower mean latency of LCDN is therefore partly an artifact of selection bias rather than a reflection of superior scheduling.

Figure~\ref{fig:tc_quality}(d) illustrates this effect directly. It shows the Cumulative Distribution Function (CDF) of per-frame TC end-to-end latency for a representative scenario (H2-B2, tree topology) in which all three commodity-switch methods admit every TC flow. The CDF plots, for each latency value on the $x$-axis, the fraction of TC frames that experienced a latency at or below that value; a curve that rises steeply and reaches 1.0 early indicates that most frames are delivered quickly. In this scenario, TNP delivers all frames within 500~µs with a mean latency of 368~µs, whereas LCDN achieves a mean of 395~µs but spreads its deliveries over a wider range up to approximately 620~µs. The steeper rise of the TNP curve confirms that its latency distribution is more tightly concentrated, which is a direct benefit of the collision-free release-time assignment: each TC frame traverses the network without encountering any contention. LCDN, by contrast, relies on priority-based queuing at switches where TC frames may experience variable queuing delays depending on the instantaneous traffic mix, resulting in a smoother and more spread-out CDF.

Overall, TNP emerges as the strongest method for TC traffic across all metrics: it admits 100\% of TC flows in every scenario and achieves the lowest mean end-to-end latency among the methods when taking into account admission. TP achieves the same perfect admission rate at a modest latency cost due to conservative worst-case BE blocking margins. LCDN's seemingly competitive latency is conditional on its lower admission rate under heavy load, which limits its applicability. All three methods operate on commodity switches without TSN hardware, in contrast to HERMES which requires full TAS and CBS support at every switch port.

\subsection{BE Throughput}
\label{sec:eval_be}

Figure~\ref{fig:be_throughput} shows the aggregate BE throughput achieved by each method across all 40 scenarios. HERMES achieves the highest BE throughput with a median of approximately 48~Mbps and a mean of approximately 50~Mbps, with individual scenarios reaching beyond 100~Mbps. This is consistent with its design philosophy: HERMES explicitly optimizes TC frame placement to maximize the time available for BE traffic on each switch port, and its per-switch TAS hardware ensures that BE frames can fill every gap in the GCL cycle without any source-side rate limiting.

\begin{figure}[t]
\centering
\includegraphics[width=0.9\columnwidth]{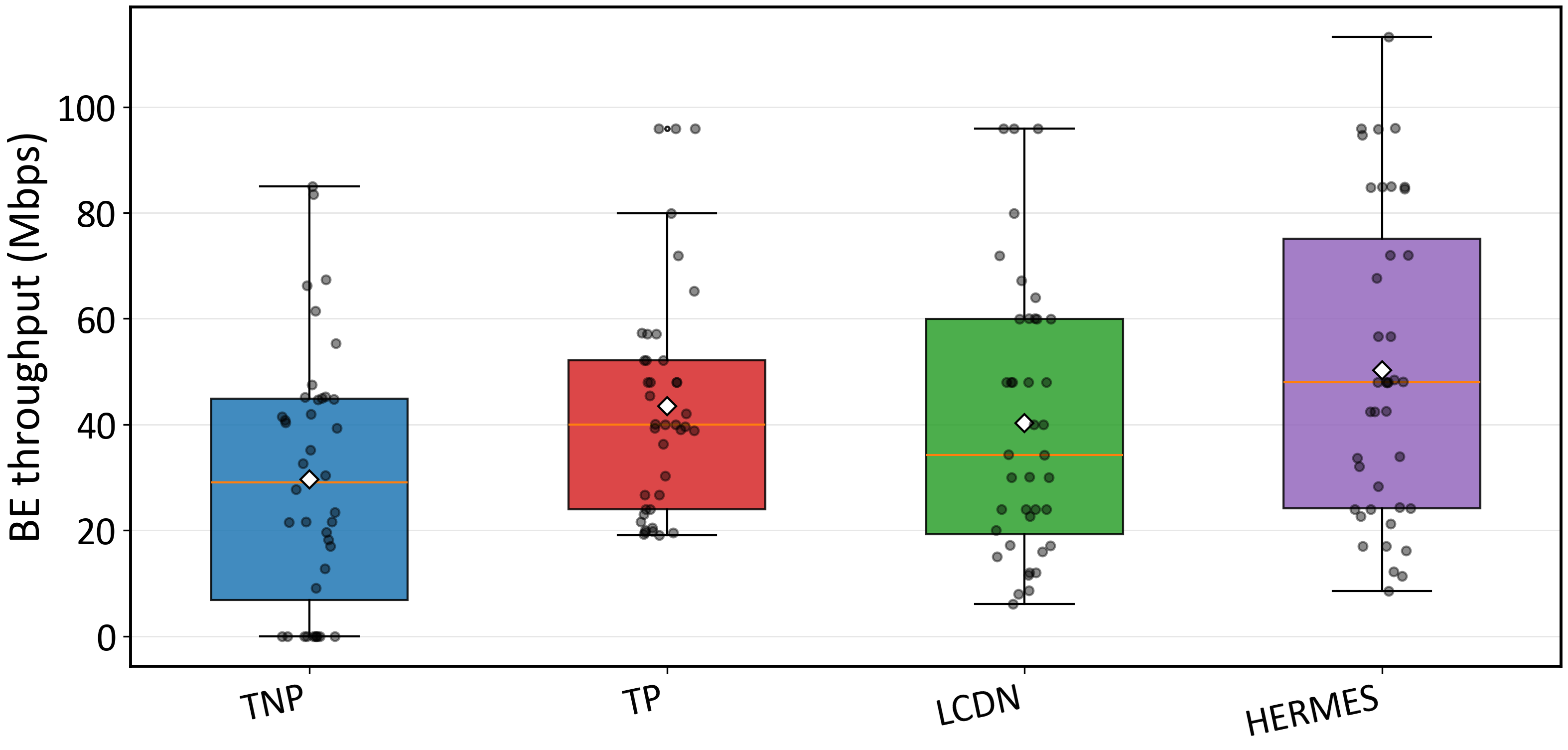}
\caption{Aggregate BE throughput across all 40 configurations.}
\label{fig:be_throughput}
\end{figure}

Among the commodity-switch methods, TP outperforms TNP in BE throughput, with a median of approximately 40~Mbps compared to approximately 29~Mbps for TNP. This difference is a direct consequence of their architectural designs. TNP is non-work-conserving: BE traffic is confined to designated BE phases within the GCL cycle and cannot transmit during TC phases or guard bands, even if the medium is idle. TP, by contrast, is work-conserving: BE endpoints may transmit whenever their CBS credit is non-negative and the egress link is not occupied by a TC frame, without waiting for a scheduled phase. This continuous access to the medium gives TP a structural advantage in BE utilization.

LCDN falls between TNP and TP, with a median of approximately 34~Mbps. Although LCDN does not impose guard bands or temporal partitioning, its token-bucket-based rate limiting constrains BE transmission more conservatively than the CBS shaping used in the TP method. As a result, LCDN cannot fully exploit the work-conserving advantage that TP achieves through its combination of strict-priority queuing and credit-based shaping.

The BE throughput results highlight a fundamental trade-off in the SbDN design. TNP provides the strongest TC latency guarantees and the tightest latency distribution, but its temporal partitioning reduces the time available for BE traffic. TP relaxes this constraint by making BE work-conserving, improving BE throughput at a small cost in TC latency. The choice between the two methods therefore depends on whether the targeted system requires to prioritize minimal TC latency or higher BE utilization.

\subsection{Overall Comparison and Trade-offs}
\label{sec:eval_tradeoff}

Figure~\ref{fig:tradeoff} summarizes the trade-offs across all four methods by plotting the mean TC admission rate against the mean BE throughput, with the circle size representing the estimated relative deployment cost of the switching infrastructure. We assign a baseline cost of $1\times$ to pure FIFO switches (used by TNP), $2\times$ to commodity switches with additional features such as strict-priority queuing (used by TP and LCDN), and $7\times$ to full TSN-capable switches with TAS and CBS hardware at every port (used by HERMES). These multipliers are approximate, order-of-magnitude estimates obtained from a survey of current commercial switch prices for representative devices in each category, and are consistent with the cost gap between TSN-capable and commodity switches reported in~\cite{Cost}. The ideal operating point lies in the top-right corner of the chart, combining the highest admission rate with the highest BE throughput, while having the smallest possible circle (lowest cost).

\begin{figure}[t]
\centering
\includegraphics[width=0.8\columnwidth]{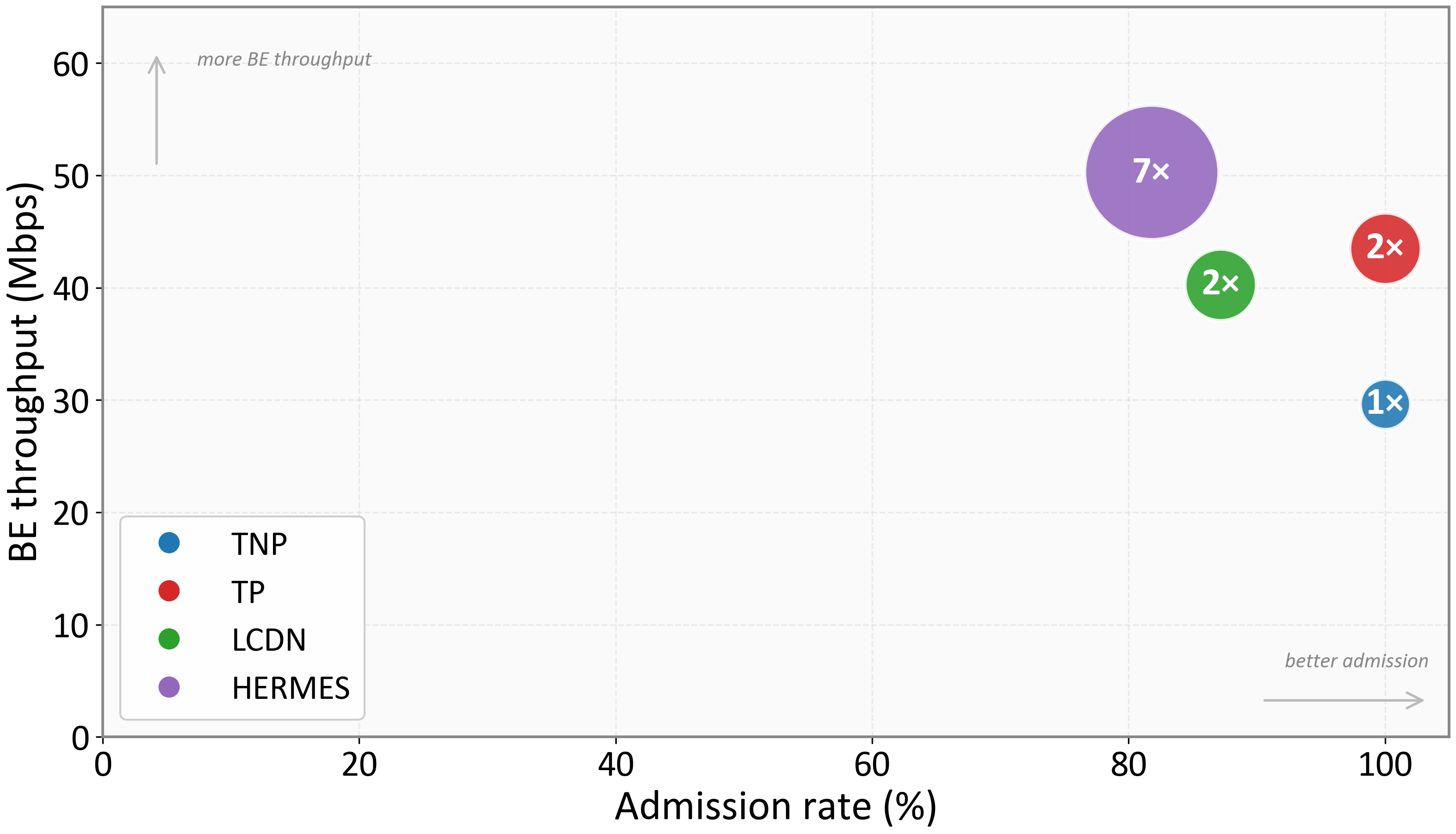}
\caption{Overall trade-off: mean TC admission rate vs.\ mean BE
throughput. Circle size indicates estimated relative deployment cost.}
\label{fig:tradeoff}
\end{figure}

TNP achieves 100\% TC admission and the lowest deployment cost, but its non-work-conserving design yields the lowest BE throughput among the four methods. TP also achieves 100\% admission and substantially improves BE throughput by making BE traffic work-conserving, at the cost of requiring strict-priority switches ($2\times$). Although TNP and TP share the same perfect admission rate in this summary view, the detailed analysis in Figure~\ref{fig:tc_e2e_zoom} showed that TNP achieves lower mean end-to-end TC latency because it does not inflate per-hop offsets with worst-case BE blocking. TNP is therefore the overall winner for TC-dominated workloads: it combines the lowest deployment cost, the highest admission rate, and the tightest latency guarantees. LCDN sits between TNP and TP on both axes, achieving slightly lower admission and moderate BE throughput at the same switch cost as TP ($2\times$). HERMES delivers the highest BE throughput, benefiting from its algorithm that explicitly optimizes for BE utilization, but at the expense of lower TC admission, significantly higher TC latency, and a deployment cost that is roughly $7\times$ that of a FIFO-based network.

The results demonstrate that source-based scheduling on commodity switches can match or exceed the TC guarantees of switch-based TSN solutions at a fraction of the deployment cost, while offering competitive BE throughput through the choice between the TNP and TP methods.

\section{Conclusion and Future Work}\label{sec:conclusions}

This paper presented SbDN, a multi-agent source-based architecture for deterministic networking that achieves TSN-grade latency guarantees on commodity Ethernet switches. By moving all scheduling intelligence to a centralized software controller and enforcing configurations exclusively at the source endpoints, SbDN eliminates the need for expensive TSN-capable switches and simplifies network reconfiguration: a new schedule can be computed and pushed to the endpoints in tens of milliseconds without disrupting ongoing traffic at the switches, since the switches are unaware of the scheduling and require no reconfiguration.

We proposed two methods that offer distinct trade-offs. The Temporal Network Partitioning (TNP) method operates on pure FIFO commodity switches and provides the strongest TC guarantees through strict temporal separation between traffic classes, achieving 100\% TC admission and the lowest end-to-end latency across all evaluated scenarios. The Traffic Prioritization (TP) method operates on commodity switches with strict-priority queuing and improves BE throughput by making BE traffic work-conserving, while maintaining the same perfect TC admission rate at a modest increase in TC latency. Both methods guarantee TC deadline satisfaction and were evaluated against HERMES and LCDN across 40 workload scenarios on two topologies (mesh, tree). The results demonstrate that source-based scheduling on commodity switches can match or exceed the TC guarantees of switch-based TSN solutions at a fraction of the deployment cost.

A limitation of our approach is that enforcement at the source endpoint requires software support at every sending device. Legacy endpoints that cannot be modified to run the source-side TAS, CBS, or release-time enforcement cannot participate in the SbDN architecture without an intermediate proxy or gateway. Addressing the integration of legacy devices is an important point of attention for practical deployment.

In future work, we plan to investigate machine learning approaches to further optimize the scheduling decisions made by the TC Scheduling Agent and the BE Shaping Agent. In particular, learning-based methods could adapt release-time assignments and idle slope configurations to observed traffic patterns, improving BE throughput while preserving TC deadline guarantees. We also plan to extend the architecture with online reconfiguration protocols that leverage the fast scheduling times demonstrated in this work to support dynamic flow admission and removal at runtime.

\end{document}